 \documentclass[11pt]{article}
 \usepackage[hmargin={3.0cm,3.0cm},vmargin={2.94cm,2.94cm}]{geometry}
\usepackage{float}
\usepackage{amsmath,amssymb}
\usepackage{amsfonts}
\usepackage{mathrsfs}
\usepackage{dsfont}

\usepackage{epic}
\usepackage{eepic}
\usepackage{graphicx}

\usepackage{pgf}
\usepackage{tikz}
\usetikzlibrary{arrows.meta}
\usetikzlibrary{decorations.markings}

\usepackage{amsmath}
\usepackage{mathtools}

\usepackage{bm}

\usepackage{empheq}

\usepackage{multirow}

\usepackage[skip=2pt,font=small]{caption}
\usepackage{hyperref}
\hypersetup{
colorlinks,
linkcolor={blue!60!black},
citecolor={blue!90!black},
urlcolor={blue!90!black}}

\usepackage{footnote}

\usepackage{cite}
\usepackage{enumerate}
\usepackage[curve]{xy}

 \usepackage{amsthm}

\usepackage{titlesec}

\titleformat{\paragraph}
{\normalfont\normalsize\bfseries}{\theparagraph}{1em}{}
\titlespacing*{\paragraph}
{0pt}{3.25ex plus 1ex minus .2ex}{1.5ex plus .2ex}

\usepackage[titletoc,title]{appendix}

\newtheorem{thm}{Theorem}[section]

\newtheorem{prop}[thm]{Proposition}
\theoremstyle{remark}
\newtheorem*{remark}{Remark}

\numberwithin{equation}{section}

\tikzset{->-/.style={decoration={
  markings,
  mark=at position .68 with {\arrow{Latex}}},postaction={decorate}}}
\tikzset{->>-/.style={decoration={
  markings,
  mark=at position .69 with {\arrow{Latex[sep=-10pt]Latex}}},postaction={decorate}}}
\tikzset{-<-/.style={decoration={
  markings,
  mark=at position .56 with {\arrow{Latex[reversed]}}},postaction={decorate}}}
\tikzset{-<<-/.style={decoration={
  markings,
  mark=at position .57 with {\arrow{Latex[reversed,sep=-10pt]Latex[reversed]}}},postaction={decorate}}}
  
\tikzset{-|-/.style={decoration={
  markings,
  mark=at position .51 with {\arrow{Bar}}},postaction={decorate}}}
\tikzset{-||-/.style={decoration={
  markings,
  mark=at position .49 with {\arrow{Bar[sep=-5pt] Bar}}},postaction={decorate}}}
\tikzset{-!-/.style={decoration={
  markings,
  mark=at position .51 with {\arrow{Tee Barb[length=4pt]}}},postaction={decorate}}}
\tikzset{-!!-/.style={decoration={
  markings,
  mark=at position .51 with {\arrow{Tee Barb[sep=1pt,length=4pt] Tee Barb[length=4pt]}}},postaction={decorate}}}

\newcommand{\ds}{\displaystyle}
\renewcommand{\author}[1]{\large\rm #1\\ \bigskip}
\newcommand{\address}[1]{{\normalsize\it #1\\}\bigskip}
\renewcommand{\title}[1]{\bigskip\bigskip\Large\bf #1\bigskip\bigskip\\}

\newcommand{\Bigpsi}[3]{\phantom{\Psi}_2 \kern -.05em
\Psi_2\left(\genfrac{}{}{0pt}{}{#1}{#2}\biggl|#3\right)}

\newcommand{\bea}{\begin{eqnarray}}
\newcommand{\eea}{\end{eqnarray}}

\newcommand{\olamh}{\hat{\olamh}}

\def\EXP{\textrm{{\large e}}}

\newcommand{\al}{{\bm{\alpha}}}
\newcommand{\bt}{{\bm{\beta}}}
\newcommand{\gm}{{\bm{\gamma}}}

\newcommand{\ccy}{x_a}
\newcommand{\ccya}{z_w}
\newcommand{\ccyb}{x}
\newcommand{\ccyc}{z_n}
\newcommand{\ccz}{y_d}
\newcommand{\ccza}{z_e}
\newcommand{\cczb}{y}
\newcommand{\cczc}{z_s}
\newcommand{\ccx}{x}
\newcommand{\cca}{y_c}
\newcommand{\ccb}{x_c}
\newcommand{\ccc}{x_d}
\newcommand{\ccd}{x_b}
\newcommand{\cce}{y_b}
\newcommand{\ccf}{y_a}

\newcommand{\ccpa}{\alpha_1}
\newcommand{\ccpb}{\alpha_2}

\newcommand{\ccqa}{\beta_1}
\newcommand{\ccqb}{\beta_2}
\newcommand{\ccra}{\gamma_1}
\newcommand{\ccrb}{\gamma_2}
\newcommand{\ccpp}{{\al}}
\newcommand{\ccqq}{{\bt}}
\newcommand{\ccrr}{{\gm}}

\renewcommand{\L}{L}
\newcommand{\LL}{L_1}
\newcommand{\M}{L_2}
\newcommand{\Lb}{L_4}
\newcommand{\Mb}{L_3}

\newcommand{\Lxa}{\L_{x^2}}
\newcommand{\Lxb}{\L_{x^1}}
\newcommand{\Lxc}{\L_{x^0}}
\newcommand{\Dxa}{\Delta_{x^2}}
\newcommand{\Dxb}{\Delta_{x^1}}
\newcommand{\Dxc}{\Delta_{x^0}}

\newcommand{\znb}{\overline{z}_n}

\newcommand{\A}[7]{A(#1;#2,#3,#4,#5;#6,#7)}
\newcommand{\B}[7]{B(#1;#2,#3,#4,#5;#6,#7)}
\newcommand{\C}[7]{C(#1;#2,#3,#4,#5;#6,#7)}

\newcommand{\D}[8]{A_{#8}(#1;#2,#3,#4,#5;#6,#7)}

\def\EXP{\textrm{{\large e}}}

\newcounter{app}
\newcounter{sapp}[app]

\begin{document}

\vglue 2cm

\begin{center}

\title{Lax matrices for lattice equations which satisfy consistency-around-a-face-centered-cube}
\author{Andrew P.~Kels}
\address{Scuola Internazionale Superiore di Studi Avanzati,\\ Via Bonomea 265, 34136 Trieste, Italy}

\end{center}

\begin{abstract}

There is a recently discovered formulation of the multidimensional consistency integrability condition for lattice equations, called consistency-around-a-face-centered-cube (CAFCC), which is applicable to equations defined on a vertex and its four nearest neighbours on the square lattice.  This paper introduces a method of deriving Lax matrices for the equations which satisfy CAFCC.  This method gives novel Lax matrices for such equations, which include previously known equations of discrete Toda-, or Laplace-type, as well as newer equations which have only appeared in the context of CAFCC.

\end{abstract}





\section{Introduction}


In modern times, one of the important characteristics that is associated to integrability of lattice equations is the property of multidimensional consistency \cite{nijhoffwalker,BobSurQuadGraphs,ABS}. One of the notable applications of multidimensional consistency is that this property (with appropriate assumptions) can be used to almost algorithmically derive Lax pairs for equations  \cite{NijhoffQ4Lax,BobSurQuadGraphs,BHQKLax,HietarintaNEWCAC}, whereas generally this is otherwise known to be a difficult problem.  Probably the most well known examples of multidimensionally consistent equations are the two-dimensional scalar lattice (quad) equations in the Adler-Bobenko-Suris (ABS) list \cite{ABS,ABS2}.  These are equations defined on four vertices of a face of the square lattice, for which multidimensional consistency takes the form of consistency-around-a-cube (CAC).   Outside of the ABS list there are also several other known types of CAC equations, including Boussinesq-type equations \cite{TongasNijhoffBoussinesq,HietarintaBoussinesq} and other multi-component equations \cite{Kels:2018qzx,KNPT2020,ZhangKampZhang}, as well as equations satisfying formulations of multidimensional consistency which are different from CAC \cite{JoshiNakazonoCUBO,Kels:2020zjn}.

This paper is concerned with a recently discovered form of multidimensional consistency, which is known as consistency-around-a-face-centered-cube (CAFCC) \cite{Kels:2020zjn}.  The CAFCC property is applicable to equations defined on a vertex and its four nearest neighbours in the square lattice, and is formulated similarly to CAC, but instead of the regular cube the equations are required to be consistent when defined on a face-centered cube (face-centered cubic unit cell).  The reason for this is that CAFCC requires to take into account additional variables on faces, and this necessarily leads to the introduction of eight corner equations, in addition to the usual six face equations, which are naturally defined on the face-centered cube.  Overall, CAFCC requires consistency of an overdetermined system of fourteen equations on the face-centered cube for eight unknowns, while in comparison CAC requires consistency of an overdetermined system of six equations on the cube for four unknowns.

The motivation for the formulation of CAFCC came from a correspondence between discrete integrable equations and integrable lattice models of statistical mechanics, where multidimensionally consistent equations are found to be equivalent to the equations for the critical point in an asymptotic (quasi-classical) expansion of the Yang-Baxter equation.  For example, when the Yang-Baxter equation takes the form of the star-triangle relation, the equations for the critical point are known to be equivalent to ABS equations \cite{Bazhanov:2007mh,Bazhanov:2010kz,Bazhanov:2016ajm,Kels:2018xge}.  Furthermore, for the star-triangle relations which are equivalent to hypergeometric beta-type integrals, there are known to be counterpart star-triangle relations for each ABS equation.  Such star-triangle relations are constructed from products of three Boltzmann weights, whose leading order asymptotics may be written as a sum of three Lagrangian functions.  On the other hand, the face-centered quad equations which satisfy CAFCC are obtained from products of four Boltzmann weights, whose leading asymptotics are given by a sum of four Lagrangian functions.   A comparison of these parallel constructions of face-centered quad equations and ABS equations from the integrals of Boltzmann weights for statistical mechanics is outlined in Figure \ref{figintro}.

\begin{figure}[htb!]
\centering

\begin{tikzpicture}[scale=1.35]

\draw[-,thick] (-2,0)--(-2,1);
\draw[-,thick] (-2,0)--(-2.87,-0.5);
\draw[-,thick] (-2,0)--(-1.13,-0.5);
\fill (-2,0) circle (1.5pt)
node[below=3.5pt]{\color{black} $c$};
\filldraw[fill=black,draw=black] (-2,1) circle (1.5pt)
node[above=3.5pt] {\color{black} $a$};
\filldraw[fill=black,draw=black] (-2.87,-0.5) circle (1.5pt)
node[left=3.5pt] {\color{black} $d$};
\filldraw[fill=black,draw=black] (-1.13,-0.5) circle (1.5pt)
node[right=3.5pt] {\color{black} $b$};

\fill (-3.2,1.65) circle (0.01pt)
node[left=1pt]{\color{black} $(a)$};

\fill (-3.2,1.65) circle (0.01pt)
node[right=1pt]{\color{black} $\int dc\, W(c,a)W(c,b)W(d,c)$};

\fill (-3.2,-1.35) circle (0.01pt)
node[left=1pt]{\color{black} $(b)$};

\fill (-3.2,-1.35) circle (0.01pt)
node[right=1pt]{\color{black} $\substack{\ds A(a,b,c,d) =\phantom{xx} \\[0.1cm] \ds L(c,a)+L(c,b)+L(d,c)}$};

\fill (1.4,0) circle (0.01pt)
node[above=1pt]{\color{black} $\substack{\ds W(a,b)=  \ds \EXP^{\frac{L(a,b)}{\hbar}+O(1)}, \\[0.1cm] \ds \hbar\to0}$};
\draw[-latex,thick] (1.4,0)--(1.4,-1.3);


\begin{scope}[xshift=130pt,yshift=5pt]

\begin{scope}[xshift=10pt]
\draw[-,thick] (-0.8,-0.8)--(0.8,0.8); \draw[-,thick] (-0.8,0.8)--(0.8,-0.8);

\fill (-0.8,-0.8) circle (1.7pt)
node[left=3.5pt]{\color{black} $c$};
\fill (0.8,0.8) circle (1.7pt)
node[right=3.5pt]{\color{black} $b$};
\fill (-0.8,0.8) circle (1.7pt)
node[left=3.5pt]{\color{black} $a$};
\fill (0.8,-0.8) circle (1.7pt)
node[right=3.5pt]{\color{black} $d$};
\fill (0,0) circle (1.7pt)
node[below=3.5pt]{\color{black} $e$};
\end{scope}

\fill (-1.5,1.5) circle (0.01pt)
node[left=1pt]{\color{black} $(a)$};

\fill (-1.5,1.5) circle (0.01pt)
node[right=1pt]{\color{black} $\int de\, W(a,e)W(b,e)W(c,e)W(d,e)$};

\fill (-1.5,-1.5) circle (0.01pt)
node[left=1pt]{\color{black} $(b)$};

\fill (-1.5,-1.5) circle (0.01pt)
node[right=1pt]{\color{black} $\substack{\ds A(e;a,b,c,d)=\phantom{xxxx} \\[0.1cm] \ds L(a,e)+L(b,e)+L(c,e)+L(d,e)}$};


\end{scope}

\begin{scope}[xshift=-55,yshift=-80pt,scale=0.9]

\draw[-,dashed,thick] (-0.8,-0.8)--(-0.8,0.8)--(0.8,0.8)--(0.8,-0.8)--(-0.8,-0.8);
\draw[-,thick] (-0.8,-0.8)--(0.8,0.8);\draw[-,thick] (-0.8,0.8)--(-0.8,-0.8)--(0.8,-0.8);

\fill (-0.8,-0.8) circle (1.9pt)
node[left=3.5pt]{\color{black} $c$};
\fill (0.8,0.8) circle (1.9pt)
node[right=3.5pt]{\color{black} $b$};
\fill (-0.8,0.8) circle (1.9pt)
node[left=3.5pt]{\color{black} $a$};
\fill (0.8,-0.8) circle (1.9pt)
node[right=3.5pt]{\color{black} $d$};

\fill (-1.4,-2.0) circle (0.01pt)
node[left=1pt]{\color{black} $(c)$};

\fill (-1.55,-1.9) circle (0.01pt)
node[right=1pt]{\color{black} $\substack{\mbox{ABS quad equation} \\[0.3cm] \ds \frac{\partial A(a,b,c,d)}{\partial c}=0}$};

\end{scope}

\begin{scope}[xshift=132pt,yshift=-80pt,scale=1.0]

\begin{scope}[xshift=10pt]
\draw[-,thick] (-0.8,-0.8)--(0.8,0.8); \draw[-,thick] (-0.8,0.8)--(0.8,-0.8);
\draw[-,dashed,thick] (-0.8,-0.8)--(-0.8,0.8)--(0.8,0.8)--(0.8,-0.8)--(-0.8,-0.8);

\fill (-0.8,-0.8) circle (1.8pt)
node[left=3.5pt]{\color{black} $c$};
\fill (0.8,0.8) circle (1.8pt)
node[right=3.5pt]{\color{black} $b$};
\fill (-0.8,0.8) circle (1.8pt)
node[left=3.5pt]{\color{black} $a$};
\fill (0.8,-0.8) circle (1.8pt)
node[right=3.5pt]{\color{black} $d$};
\fill (0,0) circle (1.8pt)
node[below=3.5pt]{\color{black} $e$};
\end{scope}

\fill (-1.5,-1.85) circle (0.01pt)
node[left=1pt]{\color{black} $(c)$};

\fill (-1.43,-1.75) circle (0.01pt)
node[right=1pt]{\color{black} $\substack{\mbox{Face-centered quad equation} \\[0.3cm] \ds \frac{\partial A(e;a,b,c,d)}{\partial e}=0}$};

\end{scope}

\end{tikzpicture}
\caption{Outline of Yang-Baxter/multidimensional consistency correspondence (for simplicity, parameter dependences not shown): (a) An integral of a product of Boltzmann weights $W(a,b)$ which satisfies a form of the Yang-Baxter equation.  (b) A sum of Lagrangian functions $L(a,b)$ which arise in an asymptotic (quasi-classical) expansion $\hbar\to0$ of (a). (c) An integrable (multidimensionally consistent) quad equation, which arises as the equation for the critical/saddle point of the integral of (a).  The combination of three Boltzmann weights on the left results in ABS quad equations, and the combination four Boltzmann weights on the right results in face-centered quad equations. 
Note that a combination of two Boltzmann weights (but not in an integral) also has an interpretation in terms of discrete integrability, namely as one of the components for two-component Yang-Baxter maps \cite{Kels:2019ktt}.}
\label{figintro}
\end{figure}
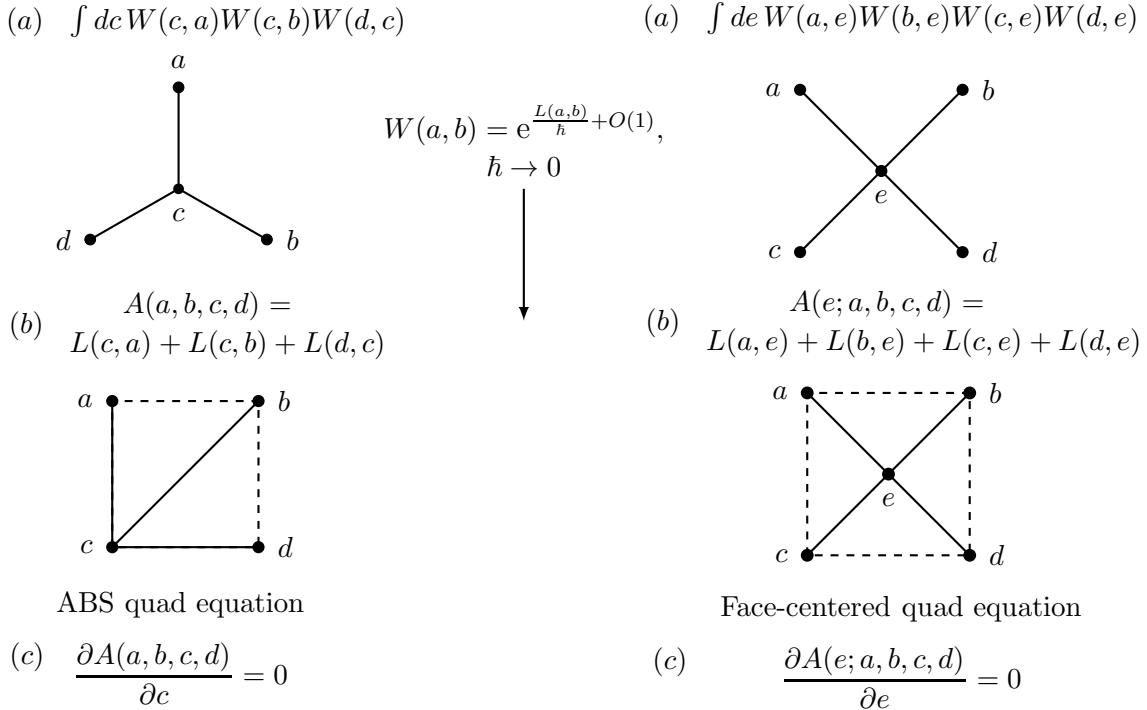

In the approach of Figure \ref{figintro}, the Boltzmann weights can be seen to be the fundamental objects from which both types of multidimensionally consistent equations can be derived.  The CAFCC equations obtained in this way include equations that have previously appeared in the literature as discrete Laplace-type equations \cite{AdlerPlanarGraphs,BobSurQuadGraphs,MR2467378,BG11,Suris_DiscreteTimeToda}, which may be identified with expressions for CAFCC equations which were denoted as type-A in \cite{Kels:2020zjn}.   However, the construction outlined in Figure \ref{figintro} is more general than that for discrete Laplace-type equations (because the equations on the right are considered independently of the ABS equations on the left), and this leads also to new equations, called type-B and type-C, that have not been previously considered outside of CAFCC.

The purpose of this paper is to present a method to derive Lax pairs for the face-centered quad equations from the property of CAFCC, in analogy with the method used to obtain Lax pairs from CAC for regular quad equations.  However, because of the differences in the formulations of CAC and CAFCC, the method of deriving Lax pairs from the former does not extend to the latter.  The main obstacle is the additional face variables for the face-centered quad equations which have no analogue for CAC, and need to be considered in an evolution around the face-centered cube.   To overcome this, a more suitable evolution on the face-centered cube will be chosen, where instead of evolving from a corner vertex to a corner vertex, the Lax matrices are obtained from different evolutions from face vertices to face vertices.  Such an alternative approach will be seen to yield the desired Lax matrices for both the type-A and type-B CAFCC equations.  The main results of this paper are both the new method to derive Lax matrices, and the resulting expressions for the Lax matrices themselves.

The paper's layout is as follows.  In Section \ref{sec:FCC}, an overview of the CAFCC property and equations will be given.  In Section \ref{sec:laxmethods}, the details will be given on how to derive the Lax matrices and compatibility conditions from the property of CAFCC, for both type-A and type-B equations.  In Section \ref{sec:laxexamples}, explicit examples of the Lax matrices that are obtained through the methods of Section \ref{sec:laxmethods} will be given, using the expressions for the CAFCC equations which are listed in  Appendix \ref{app:equations}.

\section{Face-centered quad equations and CAFCC}\label{sec:FCC}

The concept of face-centered quad equations, and their integrability in terms of multidimensional consistency formulated as consistency-around-a-face-centered-cube (CAFCC), was previously given by the author \cite{Kels:2020zjn}.
A face-centered quad equation may be written as 
\begin{align}\label{afflin}
\A{x}{x_a}{x_b}{x_c}{x_d}{\ccpp}{\ccqq}=0,
\end{align}
where $A$ is a multivariate polynomial of five variables $x,x_a,x_b,x_c,x_d$.  There is no restriction on the degree of the face variable $x$, but the expression \eqref{afflin} should be degree 1 in each of the four corner variables $x_a,x_b,x_c,x_d$. This is typically referred to as the affine-linear, or the multilinear property, which is a natural requirement for the equations to define a unique evolution in the lattice.  In \eqref{afflin}, the $\ccpp$ and $\ccqq$ are parameters which each have two components, as
\begin{align}
\label{pardefs}
\ccpp=(\ccpa,\ccpb),\qquad\ccqq=(\ccqa,\ccqb).
\end{align}

The face-centered quad equations which were found to satisfy CAFCC, also have a typical form
\begin{align}
\label{4legtemp}
\frac{a(x;x_a;\alpha_2,\beta_1)a(x;x_d;\alpha_1,\beta_2)}{a(x;x_b;\alpha_2,\beta_2)a(x;x_c;\alpha_1,\beta_1)}=1,
\end{align}
where $a(x_a;x_b;\alpha,\beta)$, is a ratio of polynomials of at most degree 1 in $x_b$.  In this form, the face-centered quad equation may be regarded as an equation on the vertices and solid edges shown in Figure \ref{fig-face}.  This is useful for a graphical presentation of the assignment of parameters to the face-centered cube for CAFCC, which will be used in the following.

 \begin{figure}[tbh]
\centering

\begin{tikzpicture}[scale=0.7]

\begin{scope}[xshift=300pt]

\draw[-,gray,very thin,dashed] (5,-1)--(5,3)--(1,3)--(1,-1)--(5,-1);
\draw[-,thick] (5,3)--(1,-1);
\draw[-,thick] (5,-1)--(1,3);
\fill (0.8,-0.0) circle (0.1pt)
node[left=0.5pt]{\color{black}\small $\ccpa$};
\fill (5.2,-0.0) circle (0.1pt)
node[right=0.5pt]{\color{black}\small $\ccpa$};
\fill (4,-1.2) circle (0.1pt)
node[below=0.5pt]{\color{black}\small $\ccqb$};
\fill (4,3.2) circle (0.1pt)
node[above=0.5pt]{\color{black}\small $\ccqb$};
\fill (0.8,2.0) circle (0.1pt)
node[left=0.5pt]{\color{black}\small $\ccpb$};
\fill (5.2,2.0) circle (0.1pt)
node[right=0.5pt]{\color{black}\small $\ccpb$};
\fill (2,-1.2) circle (0.1pt)
node[below=0.5pt]{\color{black}\small $\ccqa$};
\fill (2,3.2) circle (0.1pt)
node[above=0.5pt]{\color{black}\small $\ccqa$};
\fill (3,1) circle (4.1pt)
node[left=2.5pt]{\color{black} $\ccx$};
\fill (1,-1) circle (4.1pt)
node[left=1.5pt]{\color{black} $x_c$};
\filldraw[fill=black,draw=black] (1,3) circle (4.1pt)
node[left=1.5pt]{\color{black} $x_a$};
\fill (5,3) circle (4.1pt)
node[right=1.5pt]{\color{black} $x_b$};
\filldraw[fill=black,draw=black] (5,-1) circle (4.1pt)
node[right=1.5pt]{\color{black} $x_d$};

\draw[-,dotted,thick] (2,-1.2)--(2,3.2);\draw[-,dotted,thick] (4,-1.2)--(4,3.2);
\draw[-,dotted,thick] (0.8,-0.0)--(5.2,-0.0);\draw[-,dotted,thick] (0.8,2.0)--(5.2,2.0);

\end{scope}

\end{tikzpicture}

\caption{Variables and parameters on the vertices and edges of a face of the face-centered cube for the face-centered quad equation \eqref{4legtemp}.}  
\label{fig-face}
\end{figure}
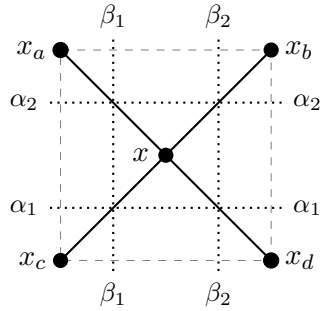

\subsection{Consistency-around-a-face-centered-cube}\label{sec:CAFCC}

CAFCC can be defined in terms of fourteen face-centered quad equations of the form \eqref{afflin}, which are centered at fourteen vertices of the two halves of the face-centered cube shown in Figure \ref{CAFCCcube}.  The face-centered cube is presented in this way to more clearly see the configurations of edges and associated parameters, while this also resembles the form of the Yang-Baxter equation from which CAFCC was originally derived.  In Figure \ref{CAFCCcube}, there are three parameters $\al$, $\bt$, $\gm$, associated to three orthogonal lattice directions, while the components of $\gm$ are always exchanged at an edge where two faces which meet orthogonally.  

The fourteen equations on the face-centered cube will be denoted here by
\begin{align}\label{14total}
    \D{x}{x_a}{x_b}{x_c}{x_d}{\al}{\bt}{i}=0,\qquad i=1,\ldots,14.
\end{align}  
Then the following six equations
\begin{align}\label{6face}
\begin{array}{rrr}
\D{\ccyb}{\ccy}{\ccd}{\ccb}{\ccc}{\ccpp}{\ccqq}{1}=0, \hspace{-0.17cm} & \D{\ccya}{\ccf}{\ccy}{\cca}{\ccb}{\ccpp}{\ccrr}{2}=0,\hspace{-0.17cm}
&\D{\ccyc}{\ccf}{\cce}{\ccy}{\ccd}{\ccrr}{\ccqq}{3}=0, \\[0.1cm]
\D{\cczb}{\ccf}{\cce}{\cca}{\ccz}{\ccpp}{\ccqq}{4}=0, \hspace{-0.17cm} & \D{\ccza}{\cce}{\ccd}{\ccz}{\ccc}{\ccpp}{\ccrr}{5}=0,\hspace{-0.17cm} &\D{\cczc}{\cca}{\ccz}{\ccb}{\ccc}{\ccrr}{\ccqq}{6}=0,
\end{array}
\end{align}
are centered at the six face vertices $\ccyb,\ccya,\ccyc,\cczb,\ccza,\cczc$, respectively, and the following eight equations
\begin{align}\label{8corner}
\begin{array}{rr}
\ds\D{\ccy}{\ccya}{\ccf}{\ccyb}{\ccyc}{(\ccqa,\ccrb)}{(\ccpb,\ccra)}{7}=0, \hspace{0.05cm}
&\D{\ccd}{\ccza}{\cce}{\ccyb}{\ccyc}{(\ccqb,\ccrb)}{(\ccpb,\ccra)}{8}=0, \\[0.1cm]
\D{\ccb}{\ccya}{\cca}{\ccyb}{\cczc}{(\ccqa,\ccrb)}{(\ccpa,\ccra)}{9}=0, \hspace{0.05cm}
&\D{\ccc}{\ccza}{\ccz}{\ccyb}{\cczc}{(\ccqb,\ccrb)}{(\ccpa,\ccra)}{10}=0, \\[0.1cm]
\D{\ccf}{\ccya}{\ccy}{\cczb}{\ccyc}{(\ccqa,\ccra)}{(\ccpb,\ccrb)}{11}=0, \hspace{0.05cm}
&\D{\cce}{\ccza}{\ccd}{\cczb}{\ccyc}{(\ccqb,\ccra)}{(\ccpb,\ccrb)}{12}=0, \\[0.1cm]
\D{\cca}{\ccya}{\ccb}{\cczb}{\cczc}{(\ccqa,\ccra)}{(\ccpa,\ccrb)}{13}=0, \hspace{0.05cm}
&\D{\ccz}{\ccza}{\ccc}{\cczb}{\cczc}{(\ccqb,\ccra)}{(\ccpa,\ccrb)}{14}=0, 
\end{array}
\end{align}
are centered at the eight corner vertices $\ccy,\ccd,\ccb,\ccc,\ccf,\cce,\cca,\ccz$, respectively, of Figure \ref{CAFCCcube}.

\begin{figure}[hbt!]
\centering
\begin{tikzpicture}

\begin{scope}[scale=0.92]

\draw[-,gray!60!white,very thin,dashed] (-3,3)--(-1,4)--(3,4)--(3,0);
\draw[-,gray!60!white,very thin,dashed] (-3,-1)--(-1,0)--(3,0)--(1,-1)--(-3,-1);

\draw[-,very thick] (-3,-1)--(3,0);\draw[-,very thick] (-1,0)--(1,-1);
\draw[-,very thick] (-3,-1)--(-1,4);\draw[-,very thick] (-3,3)--(-1,0);
\draw[-,very thick] (-1,0)--(3,4);\draw[-,very thick] (-1,4)--(3,0);
\draw[-,very thick] (-2.98,-1)--(-2.98,3);\draw[-,very thick] (-0.98,0)--(-0.98,4);
\draw[-,very thick] (-3.01,-1)--(-3.01,3);\draw[-,very thick] (-1.01,0)--(-1.01,4);
\draw[-,black,thick,dotted] (-1.9,3.75)--(-1.5,3.75)--(-1.5,-0.25)--(3.25,-0.25);
\draw[-,black,thick,dotted] (-2.5,-1.25)--(0,0)--(0,4)--(0.4,4.2);
\draw[-,black,thick,dotted] (-2.9,3.25)--(-2.5,3.25)--(-2.5,-0.75)--(2.25,-0.75);
\draw[-,black,thick,dotted] (-0.5,-1.25)--(2,0)--(2,4)--(2.4,4.2);
\draw[-,black,thick,dotted] (-3.5,0.2) .. controls (-3.25,0.4) ..(-3,1) .. controls (-2.75,1.95) .. (-2.5,2.25) .. controls (-2.2,2.45) and (-1.8,2.65) .. (-1.5,2.75) .. controls (-1.25,2.5) .. (-1,2) .. controls (-0.5,1.2) .. (0,1) -- (3.3,1); 
\draw[-,black,thick,dotted] (-3.5,1.8) .. controls (-3.25,1.6) .. (-3,1) .. controls (-2.75,0.5) .. (-2.5,0.25) .. controls (-2.2,0.25) and (-1.8,0.45) .. (-1.5,0.75) .. controls (-1.25,1) .. (-1,2) .. controls (-0.5,2.8) .. (0,3) -- (3.3,3); 

\filldraw[fill=black,draw=black] (1,-1) circle (3.1pt)
node[below=1.5pt]{\small $\ccc$};
\filldraw[fill=black,draw=black] (0,-0.5) circle (3.1pt)
node[above=1.5pt]{\small $\ccyb$};
\filldraw[fill=black,draw=black] (-1,0) circle (3.1pt);
\fill[black!] (-1,0.1) circle (0.01pt)
node[right=4pt]{\color{black}\small $\ccy$};
\filldraw[fill=black,draw=black] (-3,-1) circle (3.1pt)
node[below=1.5pt]{\small $\ccb$};
\filldraw[fill=black,draw=black] (3,0) circle (3.1pt)
node[right=1.5pt]{\color{black}\small $\ccd$};
\filldraw[fill=black,draw=black] (-3,3) circle (3.1pt)
node[left=1.5pt]{\color{black}\small $\cca$};
\filldraw[fill=black,draw=black] (-1,4) circle (3.1pt)
node[above=1.5pt]{\small $\ccf$};
\filldraw[fill=black,draw=black] (-2,1.5) circle (3.1pt);
\fill[black!] (-1.95,1.5) circle (0.01pt)
node[below=7.5pt]{\color{black}\small $\ccya$};
\filldraw[fill=black,draw=black] (1,2) circle (3.1pt)
node[above=3.5pt]{\small $\ccyc$};
\filldraw[fill=black,draw=black] (3,4) circle (3.1pt)
node[above=1.5pt]{\small $\cce$};

\draw[black] (-1.9,3.75) circle (0.01pt)
node[above=1.5pt]{\color{black}\small $\ccpb$};
\draw[black] (-2.9,3.25) circle (0.01pt)
node[above=1.5pt]{\color{black}\small $\ccpa$};
\draw[black] (-2.5,-1.35) circle (0.01pt)
node[below=1.5pt]{\color{black}\small $\ccqa$};
\draw[black] (-0.5,-1.35) circle (0.01pt)
node[below=1.5pt]{\color{black}\small $\ccqb$};
\draw[black] (-3.5,1.8) circle (0.01pt)
node[left=1.5pt]{\color{black}\small $\ccrb$};
\draw[black] (-3.5,0.2) circle (0.01pt)
node[left=1.5pt]{\color{black}\small $\ccra$};

\begin{scope}[xshift=250,yshift=0,rotate=0]

\draw[-,gray!60!white,very thin,dashed] (-3,3)--(-1,4)--(3,4)--(1,3)--(-3,3);
\draw[-,gray!60!white,very thin,dashed] (3,4)--(3,0)--(1,-1)--(-3,-1)--(-3,3);

\draw[-,very thick] (-3,-1)--(1,3);\draw[-,very thick] (-3,3)--(1,-1);
\draw[-,very thick] (1,-1)--(3,4);\draw[-,very thick] (1,3)--(3,0);
\draw[-,very thick] (-3,3)--(3,4);\draw[-,very thick] (-1,4)--(1,3);
\draw[-,very thick] (0.98,-1)--(0.98,3);\draw[-,very thick] (2.98,0)--(2.98,4);
\draw[-,very thick] (1.01,-1)--(1.01,3);\draw[-,very thick] (3.01,0)--(3.01,4);
\draw[-,black,thick,dotted] (-2.9,3.25)--(1.5,3.25)--(1.5,-0.75)--(1.9,-0.75);
\draw[-,black,thick,dotted] (-0.5,-1.25)--(0,-1)--(0,3)--(2,4)--(2.4,4.2);
\draw[-,black,thick,dotted] (-3.4,0)--(-2,0)--(0,0) .. controls (0.5,0.2) .. (1,1) .. controls (1.25,1.95) .. (1.5,2.25) .. controls (1.8,2.45) and (2.2,2.65) .. (2.5,2.75) .. controls (2.75,2.5) .. (3,2) .. controls (3.25,1.5) .. (3.5,1.25);
\draw[-,black,thick,dotted] (-1.9,3.75)--(2.5,3.75)--(2.5,-0.25)--(2.9,-0.25);
\draw[-,black,thick,dotted] (-2.5,-1.25)--(-2,-1)--(-2,3)--(0,4)--(0.4,4.2);
\draw[-,black,thick,dotted] (-3.5,2)--(-2,2)--(0,2) .. controls (0.5,1.8) .. (1,1) .. controls (1.25,0.5) .. (1.5,0.25) .. controls (1.8,0.25) and (2.2,0.45) .. (2.5,0.75) .. controls (2.75,1) .. (3,2) .. controls (3.25,2.95) .. (3.5,3.25);

\filldraw[fill=black,draw=black] (-3,3) circle (3.1pt);
\fill[black!] (-3.2,3) circle (0.01pt)
node[below=1.5pt]{\color{black}\small $\cca$};
\filldraw[fill=black,draw=black] (1,3) circle (3.1pt)
node[right=1.5pt]{\small $\ccz$};
\filldraw[fill=black,draw=black] (0,3.5) circle (3.1pt)
node[above=1.5pt]{\small $\cczb$};
\filldraw[fill=black,draw=black] (-1,4) circle (3.1pt)
node[above=1.5pt]{\small $\ccf$};
\filldraw[fill=black,draw=black] (3,4) circle (3.1pt)
node[above=1.5pt]{\small $\cce$};
\filldraw[fill=black,draw=black] (3,0) circle (3.1pt)
node[right=1.5pt]{\color{black}\small $\ccd$};
\filldraw[fill=black,draw=black] (-3,-1) circle (3.1pt)
node[left=1.5pt]{\small $\ccb$};
\filldraw[fill=black,draw=black] (-1,1) circle (3.1pt)
node[right=2.5pt]{\small $\cczc$};
\filldraw[fill=black,draw=black] (1,-1) circle (3.1pt)
node[below=1.5pt]{\small $\ccc$};
\filldraw[fill=black,draw=black] (2,1.5) circle (3.1pt)
node[right=2.5pt]{\small $\ccza$};

\draw[black] (-1.9,3.75) circle (0.01pt)
node[left=1.5pt]{\color{black}\small $\ccpb$};
\draw[black] (-2.9,3.25) circle (0.01pt)
node[left=1.5pt]{\color{black}\small $\ccpa$};
\draw[black] (-2.5,-1.35) circle (0.01pt)
node[below=1.5pt]{\color{black}\small $\ccqa$};
\draw[black] (-0.5,-1.35) circle (0.01pt)
node[below=1.5pt]{\color{black}\small $\ccqb$};
\draw[black] (-3.4,2) circle (0.01pt)
node[left=1.5pt]{\color{black}\small $\ccrb$};
\draw[black] (-3.4,0) circle (0.01pt)
node[left=1.5pt]{\color{black}\small $\ccra$};
\end{scope}

\end{scope}

\end{tikzpicture}
\caption{Labelling of vertices and edges for consistency-around-a-face-centered cube.} 
\label{CAFCCcube}
\end{figure}
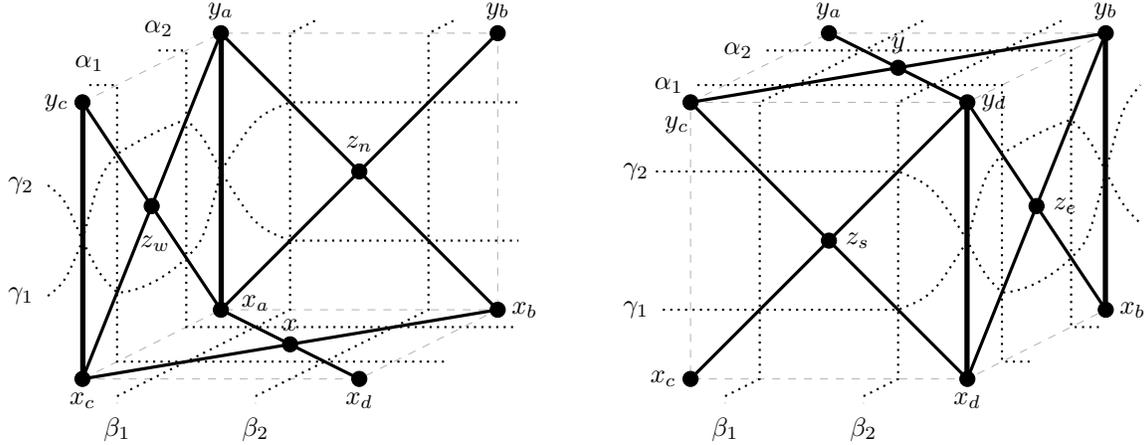

\subsubsection{CAFCC algorithm}

CAFCC for the fourteen equations \eqref{6face}, \eqref{8corner}, can be formulated as follows.

The six components of the parameters 
\begin{align}
\al=(\ccpa,\ccpb),\qquad\bt=(\ccqa,\ccqb),\qquad\gm=(\ccra,\ccrb),
\end{align}
are fixed, while 
\begin{align}
\ccyb,\ccy,\ccd,\ccb,\ccyc,\ccya,
\end{align}
are chosen as initial variables.  There remain a total of eight undetermined variables associated to the vertices of the face-centered cube, and the fourteen equations \eqref{6face}, \eqref{8corner}, to be satisfied.  

For the above initial conditions, the CAFCC property can be checked with the following six steps:

\begin{enumerate}

\item
The following two equations centered at $\ccy$ and $\ccyb$,
\begin{align}
\begin{split}
\D{\ccy}{\ccya}{\ccf}{\ccyb}{\ccyc}{(\ccqa,\ccrb)}{(\ccpb,\ccra)}{7}=0,\\
\D{\ccyb}{\ccy}{\ccd}{\ccb}{\ccc}{\ccpp}{\ccqq}{1}=0,
\end{split}
\end{align}
may be solved respectively, to uniquely determine the two variables $\ccf$, and $\ccc$.

\item
The following three equations centered at $\ccya,\ccyc,\ccf,$
\begin{align}
\begin{split}
\D{\ccya}{\ccf}{\ccy}{\cca}{\ccb}{\ccpp}{\ccrr}{2}=0, \\
\D{\ccyc}{\ccf}{\cce}{\ccy}{\ccd}{\ccrr}{\ccqq}{3}=0, \\
\D{\ccf}{\ccya}{\ccy}{\cczb}{\ccyc}{(\ccqa,\ccra)}{(\ccpb,\ccrb)}{11}=0,
\end{split}
\end{align}
may be solved respectively, to uniquely determine the three variables $\cca$, $\cce$, and $\cczb$.

\item 
For the first consistency check, both of the following two equations
\begin{align}
\begin{split}
\D{\ccd}{\ccza}{\cce}{\ccyb}{\ccyc}{(\ccqb,\ccrb)}{(\ccpb,\ccra)}{8}=0, \\
\D{\cce}{\ccza}{\ccd}{\cczb}{\ccyc}{(\ccqb,\ccra)}{(\ccpb,\ccrb)}{12}=0,
\end{split}
\end{align}
 may be used to solve for the variable $\ccza$, and the two solutions must be in agreement.

\item 
For the second consistency check, both of the following two equations
\begin{align}
\begin{split}
\D{\cca}{\ccya}{\ccb}{\cczb}{\cczc}{(\ccqa,\ccra)}{(\ccpa,\ccrb)}{13}=0,\\
\D{\ccb}{\ccya}{\cca}{\ccyb}{\cczc}{(\ccqa,\ccrb)}{(\ccpa,\ccra)}{9}=0,
\end{split}
\end{align}
 may be used to solve for the variable $\cczc$, and the two solutions must be in agreement.

\item 
For the third consistency check, each of the following four equations
\begin{align}
\begin{split}
\D{\ccza}{\cce}{\ccd}{\ccz}{\ccc}{\ccpp}{\ccrr}{5}=0, \\
\D{\cczb}{\ccf}{\cce}{\cca}{\ccz}{\ccpp}{\ccqq}{6}=0, \\
\D{\cczc}{\cca}{\ccz}{\ccb}{\ccc}{\ccrr}{\ccqq}{4}=0, \\
\D{\ccc}{\ccza}{\ccz}{\ccyb}{\cczc}{(\ccqb,\ccrb)}{(\ccpa,\ccra)}{10}=0,
\end{split}
\end{align}
may be used to solve for the final variable $\ccz$, and the four solutions must be agreement.

\item 
For the final consistency check, the remaining equation centered at $\ccz$,
\begin{align}
\begin{split}
\D{\ccz}{\ccza}{\ccc}{\cczb}{\cczc}{(\ccqb,\ccra)}{(\ccpa,\ccrb)}{14}=0,
\end{split}
\end{align}
must be satisfied by the variables that have been determined in the previous steps.

\end{enumerate}

If the above six steps are satisfied then the equations \eqref{14total} satisfy CAFCC.

\subsection{Type-A -B and -C CAFCC equations}

In addition to the CAFCC property, fifteen sets of CAFCC equations were introduced in \cite{Kels:2020zjn}, which can be grouped into two types based on the different configurations of equations on the face-centered cube.   The first type is when all of the equations of \eqref{14total} are the same, and thus may be written in terms of a single polynomial $A$, as
\begin{align}\label{CAFCCeqsA}
    \D{x}{x_a}{x_b}{x_c}{x_d}{\al}{\bt}{i}=\A{x}{x_a}{x_b}{x_c}{x_d}{\al}{\bt},\qquad i=1,\ldots,14.
\end{align}

The second type is when the equations of \eqref{14total} are given in terms of three different polynomials, as
\begin{align}\label{CAFCCeqsBC}
    \D{x}{x_a}{x_b}{x_c}{x_d}{\al}{\bt}{i}=\left\{\begin{array}{rl}
    \A{x}{x_a}{x_b}{x_c}{x_d}{\al}{\bt}, &\;\; i=2,5, \\[0.1cm]
    \B{x}{x_a}{x_b}{x_c}{x_d}{\al}{\bt}, &\;\; i=1,4,3,6, \\[0.1cm]
    \C{x}{x_a}{x_b}{x_c}{x_d}{\al}{\bt}, &\;\; i=7,\ldots,14.
    \end{array}\right.
\end{align}
Here, the face equations for $i=2,5$, are the same type-A equations which satisfy CAFCC in the form \eqref{CAFCCeqsA}, while the face equations for $i=1,4,3,6$, are denoted as type-B equations, and the remaining eight corner equations are denoted as type-C equations.  The equations which are known to satisfy CAFCC in the respective forms \eqref{CAFCCeqsA} and \eqref{CAFCCeqsBC}, are listed in Appendix \ref{app:equations}.

All three types of equations in Appendix \ref{app:equations} satisfy the symmetry (reflection on line bisecting $\bt$ edges of Figure \ref{fig-face})
\begin{align}\label{refsym}
    \D{x}{x_a}{x_b}{x_c}{x_d}{\al}{\bt}{i}=-\D{x}{x_b}{x_a}{x_d}{x_c}{\al}{\hat{\bt}}{i},
\end{align}
where $\hat{\bt}$ represents $\bt$ with the components exchanged (and similarly for $\hat{\al}$ and $\al$), {\it i.e.},
\begin{align}
    \hat{\al}=(\alpha_2,\alpha_1),\qquad\hat{\bt}=(\beta_2,\beta_1).
\end{align}
Type-A and type-B polynomials also satisfy the symmetry (reflection on line bisecting $\al$ edges of Figure \ref{fig-face})
\begin{align}\label{symAB}
   \D{x}{x_a}{x_b}{x_c}{x_d}{\al}{\bt}{i}=-\D{x}{x_c}{x_d}{x_a}{x_b}{\hat{\al}}{\bt}{i},
\end{align}
while only type-A polynomials satisfy the symmetry (reflection on $x_b$ $x_c$ diagonal of Figure \ref{fig-face})
\begin{align}\label{symA}
    \D{x}{x_a}{x_b}{x_c}{x_d}{\al}{\bt}{i}=-\D{x}{x_d}{x_b}{x_c}{x_a}{\bt}{\al}{i}.
\end{align}

For the three types of equations in \eqref{CAFCCeqsA}, \eqref{CAFCCeqsBC}, that are given in Appendix \ref{app:equations}, the four-leg expressions of the form \eqref{4legtemp} are respectively given by
\begin{align}
\label{4leg}
\mbox{Type-A: }\hspace{-0.05cm}\quad \frac{a(x;x_a;\alpha_2,\beta_1)a(x;x_d;\alpha_1,\beta_2)}{a(x;x_b;\alpha_2,\beta_2)a(x;x_c;\alpha_1,\beta_1)}=1, \\ \label{4legb}
\mbox{Type-B: }\quad \frac{b(x;x_a;\alpha_2,\beta_1)b(x;x_d;\alpha_1,\beta_2)}{b(x;x_b;\alpha_2,\beta_2)b(x;x_c;\alpha_1,\beta_1)}=1, \\ \label{4legc}
\mbox{Type-C: }\hspace{-0.02cm}\quad \frac{a(x;x_a;\alpha_2,\beta_1)c(x;x_d;\alpha_1,\beta_2)}{a(x;x_b;\alpha_2,\beta_2)c(x;x_c;\alpha_1,\beta_1)}=1,
\end{align}
where the $a(x_a;x_b;\alpha,\beta)$, $b(x_a;x_b;\alpha,\beta)$, $c(x_a;x_b;\alpha,\beta)$, are ratios of polynomials of at most degree 1 in $x_b$, and where $a(x_a;x_b;\alpha,\beta)$ additionally satisfies
\begin{align}
    a(x_a;x_b;\alpha,\beta)a(x_a;x_b;\beta,\alpha)=1.
\end{align}
Each of the expressions \eqref{4leg}--\eqref{4legc}, are related to the derivation of face-centered quad equations from the asymptotics of combinations of four Boltzmann weights which satisfy the Yang-Baxter equation, details of which were given in \cite{Kels:2020zjn}.  Through this connection, the expressions \eqref{4leg}--\eqref{4legc} may be associated to the combinations of edges and vertices shown in Figure \ref{3fig4quad}.  Note that although $b(x_a;x_b;\alpha,\beta)$ and $c(x_a;x_b;\alpha,\beta)$ are generally different, they are associated to the same double lines in Figure \ref{3fig4quad}.  This is due to a subtlety in the way these equations were originally derived from non-symmetric Lagrangian functions on edges of the face-centered cube of Figure \ref{CAFCCcube}, where the derivatives of a Lagrangian function with respect to one of the two variables associated to an edge, is different to the derivative with respect to the other variable.  The two different derivatives then give the respective expressions for $b(x_a;x_b;\alpha,\beta)$ and $c(x_a;x_b;\alpha,\beta)$.  Explicit examples of the $a(x_a;x_b;\alpha,\beta)$, $b(x_a;x_b;\alpha,\beta)$, $c(x_a;x_b;\alpha,\beta)$, for CAFCC equations, are given in Appendix \ref{app:equations}, where the expressions of \eqref{4leg} with the $a(x_a;x_b;\alpha,\beta)$ from Table \ref{table-A}, may be identified with expressions for discrete Laplace-type equations associated to type-Q ABS equations \cite{AdlerPlanarGraphs,BobSurQuadGraphs}.

\begin{figure}[htb!]
\centering
\begin{tikzpicture}[scale=0.70]

\draw[-,gray,very thin,dashed] (5,-1)--(5,3)--(1,3)--(1,-1)--(5,-1);
\draw[-,thick] (5,3)--(1,-1);
\draw[-,thick] (5,-1)--(1,3);
\fill (0.8,-0.0) circle (0.1pt)
node[left=0.5pt]{\color{black}\small $\ccpa$};
\fill (5.2,-0.0) circle (0.1pt)
node[right=0.5pt]{\color{black}\small $\ccpa$};
\fill (4,-1.2) circle (0.1pt)
node[below=0.5pt]{\color{black}\small $\ccqb$};
\fill (4,3.2) circle (0.1pt)
node[above=0.5pt]{\color{black}\small $\ccqb$};
\fill (0.8,2.0) circle (0.1pt)
node[left=0.5pt]{\color{black}\small $\ccpb$};
\fill (5.2,2.0) circle (0.1pt)
node[right=0.5pt]{\color{black}\small $\ccpb$};
\fill (2,-1.2) circle (0.1pt)
node[below=0.5pt]{\color{black}\small $\ccqa$};
\fill (2,3.2) circle (0.1pt)
node[above=0.5pt]{\color{black}\small $\ccqa$};
\fill (3,1) circle (4.1pt)
node[left=2.5pt]{\color{black} $\ccx$};
\fill (1,-1) circle (4.1pt)
node[left=1.5pt]{\color{black} $x_c$};
\filldraw[fill=black,draw=black] (1,3) circle (4.1pt)
node[left=1.5pt]{\color{black} $x_a$};
\fill (5,3) circle (4.1pt)
node[right=1.5pt]{\color{black} $x_b$};
\filldraw[fill=black,draw=black] (5,-1) circle (4.1pt)
node[right=1.5pt]{\color{black} $x_d$};

\draw[-,dotted,thick] (2,-1.2)--(2,3.2);\draw[-,dotted,thick] (4,-1.2)--(4,3.2);
\draw[-,dotted,thick] (0.8,-0.0)--(5.2,-0.0);\draw[-,dotted,thick] (0.8,2.0)--(5.2,2.0);

\fill (3,-2) circle (0.01pt)
node[below=0.5pt]{\color{black}\small $\A{\ccx}{x_a}{x_b}{x_c}{x_d}{\ccpp}{\ccqq} $};


\begin{scope}[xshift=210pt]

\draw[-,gray,very thin,dashed] (5,-1)--(5,3)--(1,3)--(1,-1)--(5,-1);
\draw[-,double,thick] (5,3)--(1,-1);
\draw[-,double,thick] (5,-1)--(1,3);
\fill (0.8,-0.0) circle (0.1pt)
node[left=0.5pt]{\color{black}\small $\ccpa$};
\fill (5.2,-0.0) circle (0.1pt)
node[right=0.5pt]{\color{black}\small $\ccpa$};
\fill (4,-1.2) circle (0.1pt)
node[below=0.5pt]{\color{black}\small $\ccqb$};
\fill (4,3.2) circle (0.1pt)
node[above=0.5pt]{\color{black}\small $\ccqb$};
\fill (0.8,2.0) circle (0.1pt)
node[left=0.5pt]{\color{black}\small $\ccpb$};
\fill (5.2,2.0) circle (0.1pt)
node[right=0.5pt]{\color{black}\small $\ccpb$};
\fill (2,-1.2) circle (0.1pt)
node[below=0.5pt]{\color{black}\small $\ccqa$};
\fill (2,3.2) circle (0.1pt)
node[above=0.5pt]{\color{black}\small $\ccqa$};
\fill (3,1) circle (4.1pt)
node[left=2.5pt]{\color{black} $\ccx$};
\fill (1,-1) circle (4.1pt)
node[left=1.5pt]{\color{black} $x_c$};
\filldraw[fill=black,draw=black] (1,3) circle (4.1pt)
node[left=1.5pt]{\color{black} $x_a$};
\fill (5,3) circle (4.1pt)
node[right=1.5pt]{\color{black} $x_b$};
\filldraw[fill=black,draw=black] (5,-1) circle (4.1pt)
node[right=1.5pt]{\color{black} $x_d$};

\draw[-,dotted,thick] (2,-1.2)--(2,3.2);\draw[-,dotted,thick] (4,-1.2)--(4,3.2);
\draw[-,dotted,thick] (0.8,-0.0)--(5.2,-0.0);\draw[-,dotted,thick] (0.8,2.0)--(5.2,2.0);

\fill (3,-2) circle (0.01pt)
node[below=0.5pt]{\color{black}\small $\B{\ccx}{x_a}{x_b}{x_c}{x_d}{\ccpp}{\ccqq} $};

\end{scope}

\begin{scope}[xshift=420pt]

\draw[-,gray,very thin,dashed] (5,-1)--(5,3)--(1,3)--(1,-1)--(5,-1);
\draw[-,thick] (5,3)--(3,1)--(1,3);
\draw[-,double,thick] (1,-1)--(3,1)--(5,-1);
\fill (0.8,-0.0) circle (0.1pt)
node[left=0.5pt]{\color{black}\small $\ccpa$};
\fill (5.2,-0.0) circle (0.1pt)
node[right=0.5pt]{\color{black}\small $\ccpa$};
\fill (4,-1.2) circle (0.1pt)
node[below=0.5pt]{\color{black}\small $\ccqb$};
\fill (4,3.2) circle (0.1pt)
node[above=0.5pt]{\color{black}\small $\ccqb$};
\fill (0.8,2.0) circle (0.1pt)
node[left=0.5pt]{\color{black}\small $\ccpb$};
\fill (5.2,2.0) circle (0.1pt)
node[right=0.5pt]{\color{black}\small $\ccpb$};
\fill (2,-1.2) circle (0.1pt)
node[below=0.5pt]{\color{black}\small $\ccqa$};
\fill (2,3.2) circle (0.1pt)
node[above=0.5pt]{\color{black}\small $\ccqa$};
\fill (3,1) circle (4.1pt)
node[left=2.5pt]{\color{black} $\ccx$};
\fill (1,-1) circle (4.1pt)
node[left=1.5pt]{\color{black} $x_c$};
\filldraw[fill=black,draw=black] (1,3) circle (4.1pt)
node[left=1.5pt]{\color{black} $x_a$};
\fill (5,3) circle (4.1pt)
node[right=1.5pt]{\color{black} $x_b$};
\filldraw[fill=black,draw=black] (5,-1) circle (4.1pt)
node[right=1.5pt]{\color{black} $x_d$};

\draw[-,dotted,thick] (2,-1.2)--(2,3.2);\draw[-,dotted,thick] (4,-1.2)--(4,3.2);
\draw[-,dotted,thick] (0.8,-0.0)--(5.2,-0.0);\draw[-,dotted,thick] (0.8,2.0)--(5.2,2.0);

\fill (3,-2) circle (0.01pt)
node[below=0.5pt]{\color{black}\small $\C{\ccx}{x_a}{x_b}{x_c}{x_d}{\ccpp}{\ccqq} $};

\end{scope}

\end{tikzpicture}
\caption{A graphical representation of  \eqref{4leg}, \eqref{4legb}, \eqref{4legc}, for the face-centered quad equations of type-A, type-B, and type-C, respectively.}
\label{3fig4quad}
\end{figure}
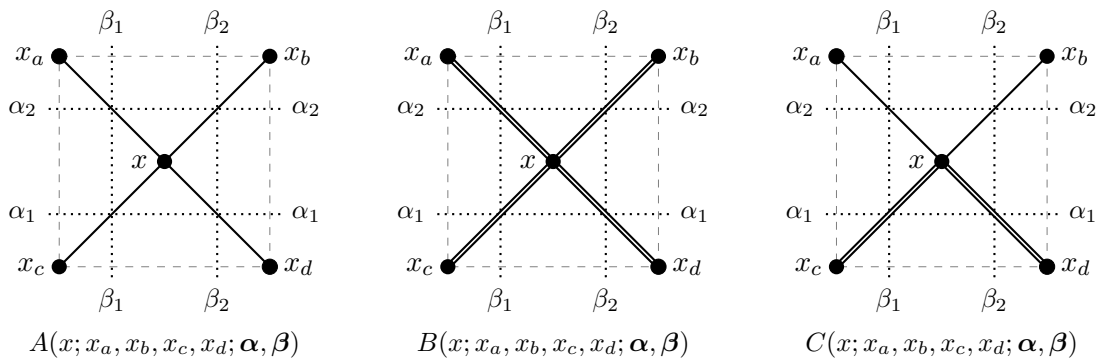


Finally, note that from \eqref{CAFCCeqsA}, and \eqref{CAFCCeqsBC}, only type-A and type-B equations appear centered at face vertices of the face-centered cube of Figure \ref{CAFCCcube}, while type-C equations only appear centered at corner vertices.  This means that only type-A or type-B equations should be regarded as systems of multidimensionally consistent equations in the lattice, while type-C equations essentially only serve as a link between type-A and type-B equations which meet at the intersection of either two or three orthogonal two-dimensional sublattices of an $n$-dimensional lattice (for $n>2$).  

\section{Lax Matrices from CAFCC}\label{sec:laxmethods}

One of the important properties of multidimensional consistency, is that a Lax pair for a multidimensional consistent equation may be derived in an almost algorithmic manner.  For regular quad equations which satisfy consistency-around-a-cube (CAC), a standard way to do this is to choose two particular vertices on the cube, {\it e.g.} those corresponding to $\ccf$ and $\ccz$, in Figure \ref{CAFCCcube}, and two different evolutions are taken from $\ccf$ to $\ccz$, namely, one that involves taking $\ccf\to\cce\to\ccz$, and another that involves taking $\ccf\to\cca\to\ccz$ (in the notation of Figure \ref{CAFCCcube}).   Due to the CAC property, the two different paths should agree for the final variable $\ccz$, and this consistency of the equations can be recast as the compatibility of appropriately defined Lax matrices.  A genuine Lax pair should be compatible only on solutions of the original CAC equation, while the resulting Lax equations are also known to be closely related to the B\"acklund transformations for the equations \cite{BobSurQuadGraphs,HietarintaNEWCAC}.

However, such a procedure as outlined above cannot be applied to CAFCC, because there is no suitable analogue for the paths that can be taken between the two vertices $\ccf$ and $\ccz$, in order to define the Lax matrices.  This is basically because the evolution on the face-centered cube also involves the eight corner equations \eqref{8corner}, which don't have a counterpart for CAC.  Besides evolution between two corner vertices, another possibility on the face-centered cube is to choose two different evolutions  between two opposite face vertices.  In this section, such an approach between face vertices will be considered, and will be seen to result in the desired compatible Lax matrices for both the type-A and type-B CAFCC equations, respectively.

\subsection{Face-centered quad equation as Lax matrix (approach 1)}\label{sec:laxmethod1}

Recall that a general face-centered quad equation \eqref{afflin} is written as
\begin{align}\label{refeq}
    \A{x}{x_a}{x_b}{x_c}{x_d}{\al}{\bt}=0.
\end{align}
The two variables $x_c$ and $x_d$, of \eqref{refeq}, will be reinterpreted here as vectors, and the equation  \eqref{refeq} itself will be reinterpreted as a matrix taking the vector associated to $x_c$, to the vector associated to $x_d$.  This may be done by first solving the equation \eqref{refeq} for $x_d$, which, due to linearity in this variable, may be written as
\begin{align}\label{xdsol}
    x_d=\frac{\A{x}{x_a}{x_b}{x_c}{0}{\al}{\bt}}
              {\sum_{j=0}^1(1-2j)\A{x}{x_a}{x_b}{x_c}{j}{\al}{\bt}}.
\end{align}
Then following the substitutions
\begin{align}\label{eqtolaxsubs}
    x_c=\frac{f}{g}, \qquad x_d=\frac{f_L}{g_L},
\end{align}
\eqref{xdsol} may be written in a matrix form
\begin{align}\label{mateq1}
    \psi_{L}=\L(x;x_a,x_b;\al,\bt)\psi,
\end{align}
where
\begin{align}
    \psi=\left(\!\!\begin{array}{c} f \\ g \end{array}\!\!\right)\!,\qquad 
    \psi_L=\left(\!\!\begin{array}{c} f_L \\ g_L \end{array}\!\!\right)\!,
\end{align}
and the $\L(x;x_a,x_b;\al,\bt)$ is a $2\times2$ matrix given by 
\begin{align}\label{Lentries}
\begin{split}
\L&(x;x_a,x_b;\al,\bt) \\
&=D_\L\!\left(\!\begin{array}{cc}
  \ds\sum_{j=0}^1(2j-1)\A{x}{x_a}{x_b}{j}{0}{\al}{\bt} &
  \ds\A{x}{x_a}{x_b}{0}{0}{\al}{\bt} \\
  \ds\sum_{j,k=0}^1(2j-1)(1-2k)\A{x}{x_a}{x_b}{j}{k}{\al}{\bt} &
  \ds\sum_{j=0}^1(1-2j)\A{x}{x_a}{x_b}{0}{j}{\al}{\bt}
 \end{array}\!\right)\!,
 \end{split}
\end{align}
where $D_\L$ is an as yet unspecified normalisation factor.

Such a reinterpretation of the face-centered quad equation \eqref{refeq} in terms of a matrix equation \eqref{mateq1}, 
can be represented diagrammatically as shown in Figure \ref{laxfig1}.  The diagram on the left of Figure \ref{laxfig1} is a direct interpretation of the construction given above for obtaining \eqref{Lentries} from \eqref{refeq}, while the diagram on the right is an equivalent interpretation which comes from using the symmetry \eqref{refsym}.  Note also that due to the symmetry \eqref{refsym}, the matrix $\L(x;x_b,x_a;\al,\hat{\bt})$ is proportional to the inverse of $\L(x;x_a,x_b;\al,\bt)$, as can be seen directly from the definition \eqref{Lentries}.  The same inverse matrix can be derived from the above procedure, by exchanging the roles of $x_d$ and $x_c$ from the beginning.   


\begin{figure}[htb!]
\centering

\begin{tikzpicture}[scale=0.7]

\fill (0,2.7) circle (0.01pt)
node[above=1pt]{$\A{x}{x_a}{x_b}{x_c}{x_d}{\al}{\bt}$};

\draw[-,very thick] (-4.0,0.0)--(4.0,0.0); \draw[-,very thick] (0.0,0.0)--(-2.0,2.0); \draw[-,very thick] (0.0,0.0)--(2.0,2.0);

\fill (-4,0) circle (3.5pt)
node[left=1pt]{$x_c$};
\fill(-4.2,-0.2) circle (0.01pt)
node[below=10pt]{$\psi$};
\fill (4,0) circle (3.5pt)
node[right=1pt]{$x_d$};
\fill(4.2,-0.2) circle (0.01pt)
node[below=10pt]{$\psi_\L$};
\fill (0,0) circle (3.5pt)
node[below=1pt]{$x$};
\fill (-2,2) circle (3.5pt)
node[left=1pt]{$x_a$};
\fill (2,2) circle (3.5pt)
node[right=1pt]{$x_b$};

\fill (0,-1.2) circle (0.01pt)
node[below=1pt]{$\L(x;x_a,x_b;\al,\bt)$};

\draw[-latex,very thick] (-3.7,-1.0)--(3.7,-1.0);




\begin{scope}[xshift=350pt]

\fill (0,2.7) circle (0.01pt)
node[above=1pt]{$\A{x}{x_b}{x_a}{x_d}{x_c}{\al}{\hat{\bt}}$};

\draw[-,very thick] (-4.0,0.0)--(4.0,0.0); \draw[-,very thick] (0.0,0.0)--(-2.0,2.0); \draw[-,very thick] (0.0,0.0)--(2.0,2.0);

\fill (-4,0) circle (3.5pt)
node[left=1pt]{$x_d$};
\fill(-4.2,-0.2) circle (0.01pt)
node[below=10pt]{$\psi_\L$};
\fill (4,0) circle (3.5pt)
node[right=1pt]{$x_c$};
\fill(4.2,-0.2) circle (0.01pt)
node[below=10pt]{$\psi$};
\fill (0,0) circle (3.5pt)
node[below=1pt]{$x$};
\fill (-2,2) circle (3.5pt)
node[left=1pt]{$x_b$};
\fill (2,2) circle (3.5pt)
node[right=1pt]{$x_a$};

\fill (0,-1.2) circle (0.01pt)
node[below=1pt]{$\L(x;x_a,x_b;\al,\bt)$};

\draw[latex-,very thick] (-3.7,-1.0)--(3.7,-1.0);

\end{scope}

\end{tikzpicture}

\caption{The face-centered quad equation \eqref{refeq}, reinterpreted as equation \eqref{mateq1} for the matrix \eqref{Lentries}.  The diagrams on the left and right are equivalent, due to the symmetry \eqref{refsym}.}  
\label{laxfig1}
\end{figure}
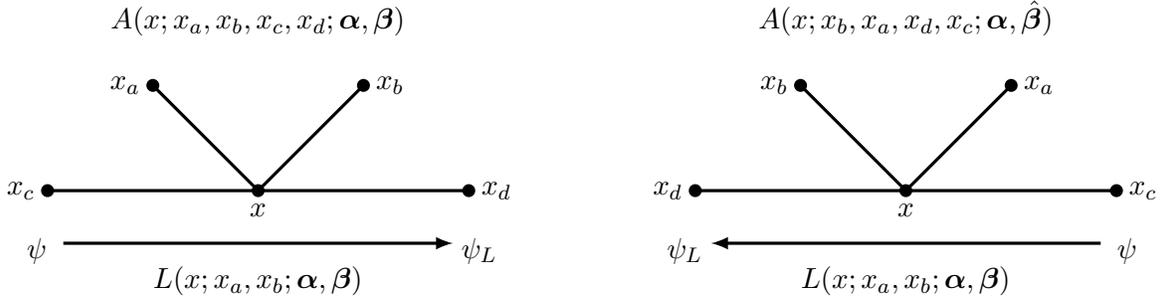

\subsubsection{CAFCC and compatibility of Lax matrices for type-A equations}

The Lax matrix \eqref{Lentries} that has been obtained from the face-centered quad equation \eqref{refeq}, can be used to define a set of compatible Lax matrices, which are derived from the CAFCC equations \eqref{14total} appearing on the face-centered cube of Figure \ref{CAFCCcube}.  As mentioned at the start of this section, this will be done by considering two different evolutions between variables on faces.  Specifically, the initial variable will be chosen as $\ccyb$ (the bottom face variable), and the final variable will be chosen as $\cczb$ (the top face variable).  Then there are four different pairs equations from \eqref{8corner} that are centered at corner variables $(x_i,y_i)$, $i\in\{a,b,c,d\}$, in Figure \ref{CAFCCcube}, which can be used to take two different evolutions for face variables, as $\ccyb\to z_k\to\cczb$, and $\ccyb\to z_l\to\cczb$, $k,l\in\{n,e,s,w\}$, $k\neq l$, respectively.  For example, the equations centered at $\ccy$ and $\ccf$ can be chosen for one evolution, taking $\ccyb\to \ccyc\to \cczb$, and the equations centered at $\ccb$ and $\cca$ can be used for the second evolution, taking $\ccyb\to \cczc\to \cczb$.

Specifically, this corresponds to the following four equations from \eqref{8corner}
\begin{align}\label{laxiniteqs}
\begin{array}{rrr}
    &\D{\ccy}{\ccya}{\ccf}{\ccyb}{\ccyc}{(\ccqa,\ccrb)}{(\ccpb,\ccra)}{7}=0,\quad 
    &\D{\ccb}{\cca}{\ccya}{\cczc}{\ccyb}{(\ccqa,\ccrb)}{(\ccra,\ccpa)}{9}=0, \\[0.1cm]
    &\D{\ccf}{\ccy}{\ccya}{\ccyc}{\cczb}{(\ccqa,\ccra)}{(\ccrb,\ccpb)}{11}=0, \quad
    &\D{\cca}{\ccya}{\ccb}{\cczb}{\cczc}{(\ccqa,\ccra)}{(\ccpa,\ccrb)}{13}=0,
\end{array}
\end{align}
which are respectively used to define four matrices for transitions associated to $\ccyb\to\ccyc$, $\ccyb\to\cczc$, $\ccyc\to\cczb$, $\cczc\to\cczb$.  The resulting matrices are respectively given in terms of \eqref{Lentries} by
\begin{align}\label{laxmatA}
\begin{array}{rrr}
&\LL=\L(\ccy;\ccya,\ccf;(\ccqa,\ccrb),(\ccpb,\ccra)), \qquad
&\M=\L(\ccb;\ccya,\cca;(\ccqa,\ccrb),(\ccpa,\ccra)), \\[0.1cm]
&\Mb=\L(\ccf;\ccy,\ccya;(\ccqa,\ccra),(\ccrb,\ccpb)), \qquad
&\Lb=\L(\cca;\ccb,\ccya;(\ccqa,\ccra),(\ccrb,\ccpa)).
\end{array}
\end{align}
The configuration of equations \eqref{laxiniteqs} and matrices \eqref{laxmatA} is shown diagrammatically in Figure \ref{LaxfigA}, in terms of the matrices and equations of Figure \ref{laxfig1}.

\begin{figure}[tbh]
\centering

\begin{tikzpicture}[scale=0.7]

\draw[-,very thick] (0,3)--(3,0)--(0,-3)--(-3,0)--(0,3);
\draw[-,very thick] (1.5,1.5)--(1.5,-1.5)--(-1.5,1.5)--(-1.5,-1.5)--(1.5,1.5);

\fill (0,3) circle (3.7pt)
node[above=1pt]{$\cczb$};
\fill (3,0) circle (3.7pt)
node[right=1pt]{$\ccyc$}
node[right=25pt]{$\psi_1$};
\fill (0,-3) circle (3.7pt)
node[below=1pt]{$\ccyb$}
node[below=15pt]{$\psi$};
\fill (-3,0) circle (3.7pt)
node[left=1pt]{$\cczc$}
node[left=25pt]{$\psi_2$};
\fill (0,0) circle (3.7pt)
node[below=5pt]{$\ccya$};
\fill (1.5,1.6) circle (0.01pt)
node[right=1pt]{$\ccf$};
\fill (1.5,1.5) circle (3.7pt);
\fill (1.5,-1.6) circle (0.01pt)
node[right=1pt]{$\ccy$};
\fill (1.5,-1.5) circle (3.7pt);
\fill (-1.5,1.6) circle (0.01pt)
node[left=1pt]{$\cca$};
\fill (-1.5,1.5) circle (3.7pt);
\fill (-1.5,-1.6) circle (0.01pt)
node[left=1pt]{$\ccb$};
\fill (-1.5,-1.5) circle (3.7pt);

\fill (-0.8,3.9) circle (0.01pt)
node[above=1pt]{$\psi_{24}$};
\fill (0.8,3.9) circle (0.01pt)
node[above=1pt]{$\psi_{13}$};

\fill (-2.4,-2.6) circle (0.01pt)
node[left=1pt]{$\M$};
\fill (2.4,-2.6) circle (0.01pt)
node[right=1pt]{$\LL$};
\fill (-2.4,2.6) circle (0.01pt)
node[left=1pt]{$\Lb$};
\fill (2.4,2.6) circle (0.01pt)
node[right=1pt]{$\Mb$};

\draw[-latex,very thick] (-0.7,-3.9)--(-4.2,-0.4);
\draw[-latex,very thick] (0.7,-3.9)--(4.2,-0.4);
\draw[-latex,very thick] (-4.2,0.4)--(-0.8,3.8);
\draw[-latex,very thick] (4.2,0.4)--(0.8,3.8);


\end{tikzpicture}

\caption{The four equations \eqref{laxiniteqs} on the face-centered cube, centered at $\ccy$, $\ccb$, $\ccf$, $\cca$, are reinterpreted in \eqref{laxmatA} as four matrices $\LL$, $\M$, $\Mb$, $\Lb$, respectively, using the graphical interpretation given in Figure \ref{laxfig1}.  The compatibility condition is $\psi_{13}\doteq\psi_{24}$, implying for the matrices $\Lb\M-\Mb\L\doteq 0$, where $\doteq$ indicates equality on solutions of the equation \eqref{laxeq} centered at $\ccya$.}  
\label{LaxfigA}
\end{figure}

Then denoting
\begin{align}
    \psi_{13}=\Mb\psi_1=\Mb\LL\psi,\qquad \psi_{24}=\Lb\psi_2=\Lb\M\psi,
\end{align}
the compatibility condition is that these two actions are consistent, {\it i.e},
\begin{align}\label{lc}
    \psi_{13}\doteq\psi_{24},
\end{align}
where $\doteq$ indicates equality on solutions of the equation centered at $\ccya$ in Figure \ref{LaxfigA}, corresponding to 
\begin{align}\label{laxeq}
    \D{\ccya}{\ccf}{\ccy}{\cca}{\ccb}{\ccpp}{\ccrr}{2}=0,
\end{align}
from \eqref{6face}.  Note that according to \eqref{CAFCCeqsA} and \eqref{CAFCCeqsBC}, \eqref{laxeq} is always a type-A equation.  Finally, in terms of the matrices \eqref{laxmatA} the compatibility condition \eqref{lc} may be written as
\begin{align}\label{laxcomp}
\Lb\M-\Mb\LL\doteq 0.
\end{align}
This final equation represents a reinterpretation of the CAFCC property of the equations \eqref{laxiniteqs}, in terms of the compatibility of the Lax matrix $\L(x;x_a,x_b;\al,\bt)$ defined in \eqref{Lentries}.  In the literature the equation of the form \eqref{laxcomp} is sometimes referred to as a Lax equation, compatibility relation, or a discrete zero curvature condition {\it e.g.} \cite{NijhoffQ4Lax,BobSurQuadGraphs,BHQKLax}.  Note that the equation \eqref{laxeq} is independent of the parameter $\ccqa$, which is a parameter of the matrices \eqref{laxmatA}.  The parameter $\ccqa$ may be regarded as (a component of) a parameter associated with the extension of the equation \eqref{laxeq} into a direction orthogonal to the two-dimensional lattice associated to the parameters $\al$ and $\gm$, and is identified as the spectral parameter in analogy with the parameter that arises in a similar way for CAC equations \cite{NijhoffQ4Lax,BobSurQuadGraphs}.

Next recall that there are three types of equations denoted as types-A, -B, and -C, which satisfy CAFCC in the two forms given in \eqref{CAFCCeqsA}, \eqref{CAFCCeqsBC}.  The construction of compatible Lax matrices that is presented above is valid for either type-A or type-C equations, because only the symmetry \eqref{refsym} has been assumed.  Note that although type-B equations also satisfy the same symmetry, they cannot be used here since there are no known sets of CAFCC equations for which the equations \eqref{laxiniteqs} are of type-B.  In terms of the graphical representation of Figure \ref{3fig4quad}, the diagram of Figure \ref{laxfig1} would be the matrix interpretation of a type-A equation, and the diagram of Figure \ref{laxfig3} below would be the matrix interpretation for a type-C equation.


\begin{figure}[htb]
\centering

\begin{tikzpicture}[scale=0.7]

\fill (0,2.7) circle (0.01pt)
node[above=1pt]{$\C{x}{x_a}{x_b}{x_c}{x_d}{\al}{\bt}$};

\draw[-,double,very thick] (-4.0,0.0)--(4.0,0.0); \draw[-,very thick] (0.0,0.0)--(-2.0,2.0); \draw[-,very thick] (0.0,0.0)--(2.0,2.0);

\fill (-4,0) circle (3.5pt)
node[left=1pt]{$x_c$};
\fill(-4.2,-0.2) circle (0.01pt)
node[below=10pt]{$\psi$};
\fill (4,0) circle (3.5pt)
node[right=1pt]{$x_d$};
\fill(4.2,-0.2) circle (0.01pt)
node[below=10pt]{$\psi_\L$};
\fill (0,0) circle (3.5pt)
node[below=1pt]{$x$};
\fill (-2,2) circle (3.5pt)
node[left=1pt]{$x_a$};
\fill (2,2) circle (3.5pt)
node[right=1pt]{$x_b$};

\fill (0,-1.2) circle (0.01pt)
node[below=1pt]{$\L(x;x_a,x_b;\al,\bt)$};

\draw[-latex,double,very thick] (-3.7,-1.0)--(3.7,-1.0);

\begin{scope}[xshift=350pt]

\fill (0,2.7) circle (0.01pt)
node[above=1pt]{$\C{x}{x_b}{x_a}{x_d}{x_c}{\al}{\hat{\bt}}$};

\draw[-,double,very thick] (-4.0,0.0)--(4.0,0.0); \draw[-,very thick] (0.0,0.0)--(-2.0,2.0); \draw[-,very thick] (0.0,0.0)--(2.0,2.0);

\fill (-4,0) circle (3.5pt)
node[left=1pt]{$x_d$};
\fill(-4.2,-0.2) circle (0.01pt)
node[below=10pt]{$\psi_\L$};
\fill (4,0) circle (3.5pt)
node[right=1pt]{$x_c$};
\fill(4.2,-0.2) circle (0.01pt)
node[below=10pt]{$\psi$};
\fill (0,0) circle (3.5pt)
node[below=1pt]{$x$};
\fill (-2,2) circle (3.5pt)
node[left=1pt]{$x_b$};
\fill (2,2) circle (3.5pt)
node[right=1pt]{$x_a$};

\fill (0,-1.2) circle (0.01pt)
node[below=1pt]{$\L(x;x_a,x_b;\al,\bt)$};

\draw[latex-,double,very thick] (-3.7,-1.0)--(3.7,-1.0);

\end{scope}

\end{tikzpicture}

\caption{The case of Figure \ref{laxfig1} for when the matrix \eqref{Lentries} is derived from a type-C equation, using the edge configuration of Figure \ref{3fig4quad}.}  
\label{laxfig3}
\end{figure}

The compatibility diagram of Figure \ref{LaxfigA} is also shown in Figure \ref{LaxfigA2}, for the type-C matrices of Figure \ref{laxfig3}.  It is seen in Figure \ref{LaxfigA2} that the double edges of Figure \ref{laxfig3} only appear on the boundary, and according to Figure \ref{3fig4quad} the central equation corresponding to \eqref{laxeq} is a type-A equation, as was the case for Figure \ref{LaxfigA}.  Thus through the approach of this subsection, type-A and type-C equations both give Lax matrices \eqref{Lentries} for type-A equations.

\begin{figure}[htb]
\centering

\begin{tikzpicture}[scale=0.7]

\draw[-,double,very thick] (0,3)--(3,0)--(0,-3)--(-3,0)--(0,3);
\draw[-,very thick] (1.5,1.5)--(1.5,-1.5)--(-1.5,1.5)--(-1.5,-1.5)--(1.5,1.5);

\fill (0,3) circle (3.7pt)
node[above=1pt]{$\cczb$};
\fill (3,0) circle (3.7pt)
node[right=1pt]{$\ccyc$}
node[right=25pt]{$\psi_1$};
\fill (0,-3) circle (3.7pt)
node[below=1pt]{$\ccyb$}
node[below=15pt]{$\psi$};
\fill (-3,0) circle (3.7pt)
node[left=1pt]{$\cczc$}
node[left=25pt]{$\psi_2$};
\fill (0,0) circle (3.7pt)
node[below=5pt]{$\ccya$};
\fill (1.5,1.6) circle (0.01pt)
node[right=1pt]{$\ccf$};
\fill (1.5,1.5) circle (3.7pt);
\fill (1.5,-1.6) circle (0.01pt)
node[right=1pt]{$\ccy$};
\fill (1.5,-1.5) circle (3.7pt);
\fill (-1.5,1.6) circle (0.01pt)
node[left=1pt]{$\cca$};
\fill (-1.5,1.5) circle (3.7pt);
\fill (-1.5,-1.6) circle (0.01pt)
node[left=1pt]{$\ccb$};
\fill (-1.5,-1.5) circle (3.7pt);

\fill (-0.8,3.9) circle (0.01pt)
node[above=1pt]{$\psi_{24}$};
\fill (0.8,3.9) circle (0.01pt)
node[above=1pt]{$\psi_{13}$};

\fill (-2.4,-2.6) circle (0.01pt)
node[left=1pt]{$\M$};
\fill (2.4,-2.6) circle (0.01pt)
node[right=1pt]{$\LL$};
\fill (-2.4,2.6) circle (0.01pt)
node[left=1pt]{$\Lb$};
\fill (2.4,2.6) circle (0.01pt)
node[right=1pt]{$\Mb$};

\draw[-latex,double,very thick] (-0.7,-3.9)--(-4.2,-0.4);
\draw[-latex,double,very thick] (0.7,-3.9)--(4.2,-0.4);
\draw[-latex,double,very thick] (-4.2,0.4)--(-0.8,3.8);
\draw[-latex,double,very thick] (4.2,0.4)--(0.8,3.8);


\end{tikzpicture}

\caption{The same compatibility condition \eqref{laxcomp} as given in Figure \ref{LaxfigA}, for the case when the Lax matrix \eqref{Lentries} is computed in terms of type-C equations and matrices of Figure \ref{laxfig3}.  Note that the equation \eqref{laxeq} centered at $\ccya$ is type-A (according to Figure \ref{3fig4quad}), as was the case for Figure \ref{LaxfigA}.}  
\label{LaxfigA2}
\end{figure}

\subsection{Face-centered quad equation as Lax matrix (approach 2)}\label{sec:laxmethod2}

Due to the non-symmetry of the type-C face-centered quad equations, a different construction of Lax matrices can be used which will lead to Lax pairs for the type-B equations, instead of Lax pairs for type-A equations which were found in Section \ref{sec:laxmethod1}.  
This is done by choosing a different pair of variables to convert into the vectors $\psi$, and $\psi_\L$.  Specifically $x_d$ and $x_b$ are used here, instead of $x_d$ and $x_c$ from Section \ref{sec:laxmethod1}.  

In the following, the type-C equation will be denoted as usual by
\begin{align}\label{refeq3}
    \C{x}{x_a}{x_b}{x_c}{x_d}{\al}{\bt}=0.
\end{align}
This type-C equation may be solved for $x_b$, which due to linearity may be written as
\begin{align}\label{xbsol}
    x_b=\frac{\C{x}{x_a}{0}{x_c}{x_d}{\al}{\bt}}
              {\sum_{j=0}^1(1-2j)\C{x}{x_a}{j}{x_c}{x_d}{\al}{\bt}}.
\end{align}
Then with the substitutions
\begin{align}\label{eqtolaxsubs2}
    x_d=\frac{f}{g},\qquad x_b=\frac{f_L}{g_L},
\end{align}
\eqref{xbsol} may be written in a matrix form
\begin{align}\label{mateq2}
    \psi_\L=\L(x;x_a,x_c;\al,\bt)\psi,
\end{align}
where
\begin{align}
    \psi=\left(\!\!\begin{array}{c} f \\ g \end{array}\!\!\right)\!,\qquad 
    \psi_\L=\left(\!\!\begin{array}{c} f_L \\ g_L \end{array}\!\!\right)\!,
\end{align}
and $\L(x;x_a,x_c;\al,\bt)$ is the $2\times2$ matrix given by
\begin{align}\label{Lentries2}
\begin{split}
\L&(x;x_a,x_c;\al,\bt) \\
  &=D_L\left(\!\!\begin{array}{cc}
  \ds\sum_{j=0}^1(2j-1)\C{x}{x_a}{0}{x_c}{j}{\al}{\bt} &
  \ds\C{x}{x_a}{0}{x_c}{0}{\al}{\bt} \\
  \ds\sum_{j,k=0}^1(2j-1)(1-2k)\C{x}{x_a}{j}{x_c}{k}{\al}{\bt} &
  \ds\sum_{j=0}^1(1-2j)\C{x}{x_a}{j}{x_c}{0}{\al}{\bt}
  \end{array}\!\!\right)\!.
\end{split}
\end{align}

Such a Lax matrix \eqref{Lentries2} constructed from type-C equations is shown diagrammatically in Figure \ref{laxfig4} ({\it c.f.} Figure \ref{laxfig3}), using the representation of the type-C equation in Figure \ref{3fig4quad}.

\begin{figure}[htb!]
\centering

\begin{tikzpicture}[scale=0.7]

\fill (0,2.7) circle (0.01pt)
node[above=1pt]{$\C{x}{x_a}{x_b}{x_c}{x_d}{\al}{\bt}$};

\draw[-,double,very thick] (-4.0,0.0)--(0.0,0.0); \draw[-,very thick] (-0.0,0.0)--(4.0,0.0); \draw[-,double,very thick] (0.0,0.0)--(-2.0,2.0); \draw[-,very thick] (0.0,0.0)--(2.0,2.0);

\fill (-4,0) circle (3.5pt)
node[left=1pt]{$x_d$};
\fill(-4.2,-0.2) circle (0.01pt)
node[below=10pt]{$\psi$};
\fill (4,0) circle (3.5pt)
node[right=1pt]{$x_b$};
\fill(4.2,-0.2) circle (0.01pt)
node[below=10pt]{$\psi_\L$};
\fill (0,0) circle (3.5pt)
node[below=1pt]{$x$};
\fill (-2,2) circle (3.5pt)
node[left=1pt]{$x_c$};
\fill (2,2) circle (3.5pt)
node[right=1pt]{$x_a$};

\fill (0,-1.2) circle (0.01pt)
node[below=1pt]{$\L(x;x_a,x_c;\al,\bt)$};

\draw[-latex,double,dashed,very thick] (-3.7,-1.0)--(3.7,-1.0);




\begin{scope}[xshift=350pt]

\fill (0,2.7) circle (0.01pt)
node[above=1pt]{$\C{x}{x_b}{x_a}{x_d}{x_c}{\al}{\hat{\bt}}$};

\draw[-,very thick] (-4.0,0.0)--(0.0,0.0); \draw[-,double,very thick] (-0.0,0.0)--(4.0,0.0); \draw[-,very thick] (0.0,0.0)--(-2.0,2.0); \draw[-,double,very thick] (0.0,0.0)--(2.0,2.0);

\fill (-4,0) circle (3.5pt)
node[left=1pt]{$x_b$};
\fill(-4.2,-0.2) circle (0.01pt)
node[below=10pt]{$\psi_\L$};
\fill (4,0) circle (3.5pt)
node[right=1pt]{$x_d$};
\fill(4.2,-0.2) circle (0.01pt)
node[below=10pt]{$\psi$};
\fill (0,0) circle (3.5pt)
node[below=1pt]{$x$};
\fill (-2,2) circle (3.5pt)
node[left=1pt]{$x_a$};
\fill (2,2) circle (3.5pt)
node[right=1pt]{$x_c$};

\fill (0,-1.2) circle (0.01pt)
node[below=1pt]{$\L(x;x_a,x_c;\al,\bt)$};

\draw[latex-,double,dashed,very thick] (-3.7,-1.0)--(3.7,-1.0);

\end{scope}

\end{tikzpicture}

\caption{The type-C equation \eqref{refeq3}, reinterpreted as equation \eqref{mateq2} for the matrix \eqref{Lentries2}.   The diagrams on the left and right are equivalent, due to the symmetry \eqref{refsym}.} 
\label{laxfig4}
\end{figure}
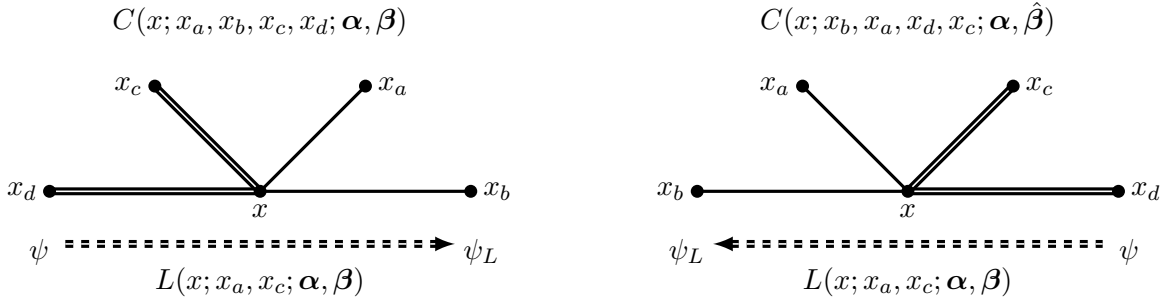

The procedure described above may be repeated to find the matrix for the inverse transition associated to $x_b\to x_d$.  Similarly to the previous case of Section \ref{sec:laxmethod1}, this results in a matrix proportional to the inverse of \eqref{Lentries2}.  However, this time the inverse matrix is not simply proportional to the matrix \eqref{Lentries2} with $x_a\leftrightarrow x_c$, because type-C equations don't satisfy the symmetry \eqref{symAB} which exchanges these variables.  Thus, for convenience the second Lax matrix will be denoted separately by $\overline{\L}(x;x_a,x_c;\al,\bt)$, where
\begin{align}\label{Lentries3}
\overline{\L}(x;x_a,x_c;\al,\bt)=L(x;x_a,x_c;\al,\bt)^{-1}.
\end{align}

Note that type-A equations satisfy the additional symmetries \eqref{symAB}, \eqref{symA}, which means that the Lax matrix \eqref{Lentries2} would be equivalent to \eqref{Lentries} up to a relabelling of the variables and parameters.  However, the type-C equations don't satisfy these additional symmetries, and the corresponding matrix \eqref{Lentries2} will thus be non-trivially different from \eqref{Lentries}.  

\subsubsection{CAFCC and compatibility of Lax matrices for type-B equations}

The compatibility of Lax matrices \eqref{Lentries2} can be derived from CAFCC in a similar way to the previous case of Section \ref{sec:laxmethod1}, but using a different choice of equations from \eqref{8corner} that are used to define the matrices \eqref{Lentries2}.  Specifically, the equations centered at $\ccy$ and $\ccf$ will again be used for one evolution, but this time they will be used to take $\ccyb\to \ccya\to \cczb$ (instead of $\ccyb\to\ccyc\to\cczb$), and for the second evolution the equations centered at $\ccd$ and $\cce$ will be used (instead of equations at $\ccb$, $\cca$) to take $\ccyb\to \ccza\to \cczb$.

Specifically, this corresponds to the following four equations from \eqref{8corner} and \eqref{CAFCCeqsBC}
\begin{align}\label{laxinit2}
\begin{array}{rrr}
&\C{\ccd}{\cce}{\ccza}{\ccyc}{\ccyb}{(\ccqb,\ccrb)}{(\ccra,\ccpb)}=0,\quad
&\C{\ccy}{\ccf}{\ccya}{\ccyc}{\ccyb}{(\ccqa,\ccrb)}{(\ccra,\ccpb)}=0, \\[0.1cm]
&\C{\cce}{\ccd}{\ccza}{\ccyc}{\cczb}{(\ccqb,\ccra)}{(\ccrb,\ccpb)}=0,\quad
&\C{\ccf}{\ccy}{\ccya}{\ccyc}{\cczb}{(\ccqa,\ccra)}{(\ccrb,\ccpb)}=0,
\end{array}
\end{align}
being respectively reinterpreted as four matrices for transitions associated to $\ccyb\to\ccza$, $\ccyb\to\ccya$, $\ccza\to\cczb$, $\ccya\to\cczb$.  The resulting matrices are respectively given in terms of \eqref{Lentries2} and \eqref{Lentries3}, by
\begin{align}\label{laxmatb}
\begin{array}{rrr}
&\LL=\L(\ccd;\cce,\ccyc;(\ccqb,\ccrb),(\ccra,\ccpb)), \qquad
&\M=\L(\ccy;\ccf,\ccyc;(\ccqa,\ccrb),(\ccra,\ccpb)), \\[0.1cm]
&\Mb=\overline{\L}(\cce;\ccd,\ccyc;(\ccqb,\ccra),(\ccrb,\ccpb)), \qquad
&\Lb=\overline{\L}(\ccf;\ccy,\ccyc;(\ccqa,\ccra),(\ccrb,\ccpb)).
\end{array}
\end{align}
The configuration of equations \eqref{laxinit2} and matrices \eqref{laxmatb} is shown diagrammatically in Figure \ref{LaxfigC}.  Note that according to Figure \ref{3fig4quad}, the equation centered at $\ccyc$ is of type-B.

\begin{figure}[tbh]
\centering

\begin{tikzpicture}[scale=0.7]

\draw[-,very thick] (-1.5,-1.5)--(-1.5,1.5)--(-3,0)--(-1.5,-1.5);\draw[-,very thick] (1.5,-1.5)--(1.5,1.5)--(3,0)--(1.5,-1.5);
\draw[-,double,very thick] (1.5,-1.5)--(-1.5,1.5)--(0,3)--(1.5,1.5)--(-1.5,-1.5)--(0,-3)--(1.5,-1.5);

\fill (0,3) circle (3.7pt)
node[above=1pt]{$\cczb$};
\fill (3,0) circle (3.7pt)
node[right=1pt]{$\ccza$}
node[right=25pt]{$\psi_1$};
\fill (0,-3) circle (3.7pt)
node[below=1pt]{$\ccyb$}
node[below=15pt]{$\psi$};
\fill (-3,0) circle (3.7pt)
node[left=1pt]{$\ccya$}
node[left=25pt]{$\psi_2$};
\fill (0,0) circle (3.7pt)
node[below=5pt]{$\ccyc$};
\fill (1.5,1.6) circle (0.01pt)
node[right=1pt]{$\cce$};
\fill (1.5,1.5) circle (3.7pt);
\fill (1.5,-1.6) circle (0.01pt)
node[right=1pt]{$\ccd$};
\fill (1.5,-1.5) circle (3.7pt);
\fill (-1.5,1.6) circle (0.01pt)
node[left=1pt]{$\ccf$};
\fill (-1.5,1.5) circle (3.7pt);
\fill (-1.5,-1.6) circle (0.01pt)
node[left=1pt]{$\ccy$};
\fill (-1.5,-1.5) circle (3.7pt);

\fill (-0.8,3.9) circle (0.01pt)
node[above=1pt]{$\psi_{24}$};
\fill (0.8,3.9) circle (0.01pt)
node[above=1pt]{$\psi_{13}$};

\fill (-2.4,-2.6) circle (0.01pt)
node[left=1pt]{$\M$};
\fill (2.4,-2.6) circle (0.01pt)
node[right=1pt]{$\LL$};
\fill (-2.4,2.6) circle (0.01pt)
node[left=1pt]{$\Lb$};
\fill (2.4,2.6) circle (0.01pt)
node[right=1pt]{$\Mb$};

\draw[-latex,double,dashed,very thick] (-0.7,-3.9)--(-4.2,-0.4);
\draw[-latex,double,dashed,very thick] (0.7,-3.9)--(4.2,-0.4);
\draw[-latex,double,dashed,very thick] (-4.2,0.4)--(-0.8,3.8);
\draw[-latex,double,dashed,very thick] (4.2,0.4)--(0.8,3.8);


\end{tikzpicture}

\caption{The four type-C equations \eqref{laxinit2} on the face-centered cube, centered at $\ccd$, $\ccy$, $\cce$, $\ccf$, are reinterpreted  in \eqref{laxmatb} as four matrices $\LL$, $\M$, $\Mb$, $\Lb$, respectively, using the graphical interpretation given in Figure \ref{3fig4quad}.  The compatibility condition is $\psi_{13}\doteq\psi_{24}$, implying for the matrices $\Lb\M-\Mb\LL\doteq 0$, where $\doteq$ indicates equality on solutions of the equation \eqref{laxeq2} centered at $\ccyc$.  Note that in terms of Figure \ref{3fig4quad}, this equation centered at $\ccyc$ is of type-B.}
\label{LaxfigC}
\end{figure}

Then denoting as before
\begin{align}
    \psi_{13}=\Mb\psi_1=\Mb\LL\psi,\qquad \psi_{24}=\Lb\psi_2=\Lb\M\psi,
\end{align}
the compatibility condition is
\begin{align}
    \psi_{13}\doteq\psi_{24},
\end{align}
where  $\doteq$ indicates equality on solutions of the equation centered at $\ccyc$ in Figure \ref{LaxfigC}, corresponding to
\begin{align}\label{laxeq2}
    \B{\ccyc}{\ccf}{\cce}{\ccy}{\ccd}{\ccrr}{\ccqq}=0,
\end{align}
from \eqref{6face} and \eqref{CAFCCeqsBC}.  In terms of the Lax matrices \eqref{laxmatb}, this compatibility condition is
\begin{align}\label{laxcomp2}
\Lb\M-\Mb\LL\doteq 0.
\end{align}
This is the final equation which reinterprets the CAFCC property of equations \eqref{laxinit2}, in terms of the compatibility of the Lax matrix defined in \eqref{Lentries2} (and its inverse \eqref{Lentries3}).  In this case, the parameter $\ccpb$ is identified as the spectral parameter resulting from extending the equation \eqref{laxeq2} into a third lattice direction.

\section{Expressions for CAFCC Lax matrices}\label{sec:laxexamples}


Using the two approaches of Section \ref{sec:laxmethods}, the expressions for compatible Lax matrices \eqref{Lentries} and \eqref{Lentries2}, for type-A and type-B equations respectively, will be given here for the face-centered quad equations listed in Appendix \ref{app:equations}.  However, there still remains to determine a choice for the normalisation factor $D_\L$, which will be done based on the expressions for the determinants of the respective Lax matrices \eqref{Lentries} and \eqref{Lentries2}.  Note that while it is possible to fix the normalisation of these matrices such that $\det(\L) =1$, this typically leads to normalisation factors $D_\L$ which involve factors of square roots, which makes the analysis for the Lax equations relatively more complicated.  This is analogous to the situation for Lax matrices derived from consistency-around-a-cube, which are also sometimes found to involve square root factors {\it e.g.} the cases of $Q4$ \cite{NijhoffQ4Lax} and each of $Q3_{(\delta=1)}$, $Q2$, $H3$, $H2$ \cite{BHQKLax}.

Fortunately, for most cases of face-centered quad equations that are listed in Appendix \ref{app:equations}, there can be found some simpler choices of normalisation of the Lax matrices which are given in a polynomial form.  Hence it is desirable to use such normalisations where possible over the irrational normalisations given by $\det(\L)=1$, and this section will focus on such cases.  The strategy for obtaining the normalisations is motivated by an observation of Bridgman {\it et.~al.} \cite{BHQKLax} for Lax matrices derived from CAC, where if a combination of determinants of the matrices can be factored into a certain ratio of polynomials, then sometimes these polynomials can be taken as valid choices of normalisation to arrive at a compatible Lax pair.

For the cases of CAFCC in this section, the determinants of the matrices \eqref{Lentries} or \eqref{Lentries2} can typically be factored into irreducible polynomials of both the variables and parameters, and the different combinations of these polynomials are used as candidates for the normalisation.  Here the simplest form of the normalisation  which results in a compatible Lax pair is chosen, which is sometimes simply $D_\L=1$, or otherwise a combination of several factors in the expressions for the (unnormalised) determinants.  With this strategy, a simpler normalisation than $\det(\L)=1$ was found for each of the equations in Appendix \ref{app:equations}, with the exception of the four cases of $A3_{(\delta=1)}$, $A2_{(\delta_1=1;\,\delta_2=1)}$, $B3_{(\delta_1=\frac{1}{2};\,\delta_2=\frac{1}{2};\,\delta_3=0)}$, and $B2_{(\delta_1=1;\,\delta_2=1;\,\delta_3=0)}$. This situation is similar to the cases of Lax pairs for ABS equations \cite{BHQKLax}, where this special choice of normalisation doesn’t work for more complicated equations such as $Q2$, $Q3_{(\delta=1)}$, and $H3$.  It would be interesting to understand why this special choice of normalisation only works for some equations, and whether it could be modified to extend to all equations including the more complicated examples.   Note that there is also another type-A equation called $A4$ \cite{Kels:2020zjn}, which isn't considered explicitly in this paper due to its complicated expression.

\subsection{Lax matrices for type-A equations (approach 1)}\label{sec:LaxA}

Recall from Section \ref{sec:laxmethod1} that there are two types of Lax matrices of the form \eqref{Lentries}, for type-A equations \eqref{laxeq}.  The first type of Lax matrix is constructed from the type-A equations themselves, and the second type of Lax matrix is constructed from type-C equations.

The face-centered quad equations of Appendix \ref{app:equations} are at most quadratic in the face variable $x$,  and the Lax matrix \eqref{Lentries} may thus be written in terms of coefficients of $x^i$, $i=0,1,2$, as
\begin{align}\label{Lentriesex}
    L(x;x_a,x_b;\al,\bt)=D_\L\bigl(x^2\Lxa+x\Lxb+\Lxc+\delta(x^2\Dxa+x\Dxb+\Dxc)\bigr).
\end{align}
Here $\delta$ is a parameter appearing for either a type-A or type-C equation, and each of the $\Lxa(x_a,x_b;\al,\bt)$, $\Lxb(x_a,x_b;\al,\bt)$, $\Lxc(x_a,x_b;\al,\bt)$, $\Dxa(x_a,x_b;\al,\bt)$, $\Dxb(x_a,x_b;\al,\bt)$, and $\Dxc(x_a,x_b;\al,\bt)$ are $2\times2$ matrices.  In the following, the Lax matrix \eqref{Lentriesex} (equivalently \eqref{Lentries}) will be given by specifying the six matrices $\Lxa,\Lxb,\Lxc$, $\Dxa,\Dxb,\Dxc$, as well as a valid normalisation factor $D_L$.

\subsubsection{\texorpdfstring{$A3_{(\delta)}$}{A3(delta)}}

The type-A equation is given by $A3_{(\delta)}$ in \eqref{a3d}.  For the equation \eqref{a3d}, the Lax matrix \eqref{Lentries} may be written in the form \eqref{Lentriesex}, where
\begin{gather}\label{a3mata}
    \Lxa=4\alpha_1\alpha_2\beta_1\beta_2\left(\!\!\begin{array}{cc}
    \tfrac{\beta_1}{\alpha_1}-\tfrac{\alpha_1}{\beta_1} & \bigl(\tfrac{\alpha_2}{\beta_1}-\tfrac{\beta_1}{\alpha_2}\bigr) x_a + \bigl(\tfrac{\beta_2}{\alpha_2}-\tfrac{\alpha_2}{\beta_2}\bigr) x_b \\[0.15cm] 0 & \tfrac{\beta_2}{\alpha_1}-\tfrac{\alpha_1}{\beta_2}
    \end{array}\!\!\right)\!,
\\[0.25cm]
    \Lxb=4\alpha_1\alpha_2\beta_1\beta_2\left(\!\!\begin{array}{cc}
    \bigl(\tfrac{\alpha_1}{\alpha_2}-\tfrac{\alpha_2}{\alpha_1}\bigr)x_a + \bigl(\tfrac{\alpha_1 \alpha_2}{\beta_1\beta_2} - \tfrac{\beta_1 \beta_2}{\alpha_1\alpha_2}\bigr)x_b & \bigl(\tfrac{\beta_1}{\beta_2}-\tfrac{\beta_2}{\beta_1})x_a x_b \\[0.15cm] \tfrac{\beta_1}{\beta_2}-\tfrac{\beta_2}{\beta_1} & \bigl(\tfrac{\alpha_1 \alpha_2}{\beta_1\beta_2}-\tfrac{\beta_1 \beta_2}{\alpha_1\alpha_2}\bigr)x_a + \bigl(\tfrac{\alpha_1}{\alpha_2}-\tfrac{\alpha_2}{\alpha_1}\bigr)x_b
    \end{array}\!\!\right)\!,
\end{gather} 
\begin{gather}\label{a3matab}
    \Lxc=4\alpha_1\alpha_2\beta_1\beta_2\left(\!\!\begin{array}{cc}
    \bigl(\tfrac{\beta_2}{\alpha_1}-\tfrac{\alpha_1}{\beta_2}\bigr)x_a x_b & 0 \\[0.15cm] \bigl(\tfrac{\beta_2}{\alpha_2}-\tfrac{\alpha_2}{\beta_2}\bigr)x_a + \bigl(\tfrac{\alpha_2}{\beta_1}-\tfrac{\beta_1}{\alpha_2}\bigr)x_b & \bigl(\tfrac{\beta_1}{\alpha_1}-\tfrac{\alpha_1}{\beta_1}\bigr)x_a x_b
    \end{array}\!\!\right)\!,
\\[0.3cm]
    \Dxa=0,\qquad\Dxb=\left(\!\!\begin{array}{cc}
    0 & (\alpha_2^2-\alpha_1^2)(\beta_1^2-\beta_2^2)(\tfrac{\alpha_1\alpha_2}{\beta_1\beta_2} - \tfrac{\beta_1\beta_2}{\alpha_1\alpha_2}) \\ 0 & 0
    \end{array}\!\!\right)\!,
\\[0.3cm]
\label{a3matb}
    \Dxc=\left(\!\!\begin{array}{cc}
    \frac{(\alpha_2^2-\beta_1^2)(\alpha_2^2-\beta_2^2)(\beta_2^2-\alpha_1^2)}{\alpha_2 \beta_2}  & \frac{(\alpha_1^2-\beta_1^2)(\alpha_1^2-\beta_2^2)\bigl(\beta_1(\alpha_2^2-\beta_2^2)x_a+\beta_2(\beta_1^2-\alpha_2^2)x_b\bigr)}{\alpha_1 \beta_1 \beta_2} \\[0.15cm] 
    0 & \frac{(\alpha_2^2-\beta_1^2)(\beta_1^2-\alpha_1^2)(\alpha_2^2-\beta_2^2)}{\alpha_2 \beta_1}
    \end{array}\!\!\right)\!.
\end{gather}
For this case, $\delta$ in \eqref{Lentries3} is identified with the parameter $\delta$ from the equation $A3_{(\delta)}$.

The determinant of \eqref{Lentriesex} with \eqref{a3mata}--\eqref{a3matb} is given by
\begin{align}
\begin{split}
\det(\L)=16D_\L^2(\alpha_1+\beta_1)(\alpha_1-\beta_1)
\bigl((\beta_1 x-\alpha_2 x_a)(\beta_1 x_a-\alpha_2 x)-\delta\tfrac{\alpha_2\beta_1}{4}(\tfrac{\alpha_2}{\beta_1}-\tfrac{\beta_1}{\alpha_2})^2\bigr) \phantom{.} \\
\times(\alpha_1+\beta_2)(\alpha_1-\beta_2)\bigl((\beta_2 x-\alpha_2 x_b)(\beta_2 x_b-\alpha_2x)-\delta\tfrac{\alpha_2\beta_2}{4}(\tfrac{\alpha_2}{\beta_2}-\tfrac{\beta_2}{\alpha_2})^2\bigr).
\end{split}
\end{align}
Four choices of factors on the right hand side of the above will be used for the normalisation factor $D_\L$ to give compatible Lax matrices as follows.

\begin{prop}\label{prop:a30}

For the case $\delta=0$, the Lax matrix \eqref{Lentriesex} defined with \eqref{a3mata}--\eqref{a3matab}, and the normalisation chosen as either
\begin{align}\label{a3norm1}
    D_L=\bigl((\alpha_1- \varepsilon \beta_1)(\alpha_1+ \varepsilon \beta_2)(\alpha_2 x - \beta_1 x_a)(\beta_2 x - \alpha_2 x_b)\bigr)^{-1},\qquad \varepsilon=\pm1,\quad \mbox{or}
\\
\label{a3norm2}
    D_L=\bigl((\alpha_1- \varepsilon \beta_1)(\alpha_1+ \varepsilon \beta_2)(\alpha_2 x - \beta_2 x_b)(\beta_1 x - \alpha_2 x_a)\bigr)^{-1},\qquad \varepsilon=\pm1,\phantom{\quad \mbox{or}}
\end{align}
satisfies the Lax equation \eqref{laxcomp} on solutions of $A3_{(\delta)}(\ccya;\ccf,\ccy,\cca,\ccb;\al,\gm)=0$.  

\end{prop}

\begin{proof}

Using the definitions \eqref{laxmatA}, the Lax equation of Proposition \ref{prop:a30} for the first normalisation \eqref{a3norm1} may be written as
\begin{align}
    \Lb\M-\Mb\LL=\frac{\frac{16\alpha_1\alpha_2\gamma_1\gamma_2}{(\beta_1 + \varepsilon \gamma_1) (\beta_1 - \varepsilon \gamma_2)} A3_{(\delta)}(\ccya;\ccf,\ccy,\cca,\ccb;\al,\gm)\left(\!\!\begin{array}{cc}
    -\gamma_1 \gamma_2 \ccya & \beta_1 \gamma_1 \ccya^2 \\ -\beta_1 \gamma_2 & \beta_1^2 \ccya 
    \end{array}\!\!\right)}
    { (\alpha_1 \ccya - \gamma_2 \ccb) (\alpha_2 \ccya - \gamma_2 \ccy) (\gamma_1 \ccya - \alpha_2 \ccf) (\gamma_1 \ccya - \alpha_1 \cca) }. 
\end{align}
A similar computation with the second normalisation \eqref{a3norm2} instead of \eqref{a3norm1} gives
\begin{align}
    \Lb\M-\Mb\LL=\frac{\frac{16\alpha_1\alpha_2\gamma_1\gamma_2}{(\beta_1 + \varepsilon \gamma_1) (\beta_1 - \varepsilon \gamma_2)} A3_{(\delta)}(\ccya;\ccf,\ccy,\cca,\ccb;\al,\gm)\left(\!\!\begin{array}{cc}
    -\beta_1^2 \ccya & \beta_1 \gamma_2 \ccya^2 \\ -\beta_1 \gamma_1 &  \gamma_1 \gamma_2 \ccya
    \end{array}\!\!\right)}
    { (\alpha_1 \ccya - \gamma_1 \cca) (\alpha_2 \ccya - \gamma_1 \ccf) (\gamma_2 \ccya - \alpha_2 \ccy) (\gamma_2 \ccya - \alpha_1 \ccb) }. 
\end{align}

\end{proof}

Another Lax matrix for $A3_{(\delta)}$ may be obtained from the type-C equation $C3_{(\delta_1;\,\delta_2;\,\delta_3)}$, given in \eqref{c3ddd}.  For \eqref{c3ddd}, the Lax matrix \eqref{Lentries} may be written in the form \eqref{Lentriesex}, where
\begin{align}\label{a3mataC}
    \Lxa=-\alpha_2\left(\!\!\begin{array}{cc}
    \beta_2 &\hspace{-0.15cm} 0 \\ 0 &\hspace{-0.15cm} \beta_1
    \end{array}\!\!\right)\!,
\;\;
    \Lxb=\left(\!\!\begin{array}{cc}
    \beta_1 \beta_2 x_a + \alpha_2^2  x_b &\hspace{-0.15cm} 0 \\ 0 &\hspace{-0.15cm} \alpha_2^2  x_a + \beta_1 \beta_2 x_b
    \end{array}\!\!\right)\!,
\;\;
    \Lxc=-\alpha_2 x_a x_b\left(\!\!\begin{array}{cc}
    \beta_1 &\hspace{-0.15cm} 0 \\ 0 &\hspace{-0.15cm} \beta_2
    \end{array}\!\!\right)\!,
\end{align}
\begin{gather}
    \Dxa=\frac{2\delta_3\beta_1\beta_2}{\alpha_1}\left(\!\!\begin{array}{cc}
    0 & \alpha_2^2 \bigl(\tfrac{x_a}{\beta_1} - \tfrac{x_b}{\beta_2}\bigr) + (\beta_2 x_b - \beta_1 x_a) \\ 0 & 0
    \end{array}\!\!\right)\!,
\\[0.3cm]
    \Dxb=\frac{2(\beta_1^2-\beta_2^2)\alpha_2}{\alpha_1}\left(\!\!\begin{array}{cc}
    0 & {\delta_3}x_a x_b - \tfrac{\alpha_1^2}{2\beta_1\beta_2} \\
    {\delta_2} & 0
    \end{array}\!\!\right)\!,
\\[0.3cm]
\label{a3matbC}
    \Dxc=\left(\!\!\begin{array}{cc}
    - {\delta_2}(\alpha_2^2 - \beta_1^2)(\tfrac{\alpha_2}{\beta_2} - \tfrac{\beta_2}{\alpha_2}) &
    \alpha_1\bigl(\beta_1 x_b - \beta_2 x_a + \alpha_2^2(\tfrac{x_a}{\beta_2} - \tfrac{x_b}{\beta_1})\bigr) \\[0.1cm]
    \tfrac{2 {\delta_2}}{\alpha_1}\bigl(\beta_1 \beta_2(\beta_2 x_a - \beta_1 x_b) + \alpha_2^2 (\beta_2 x_b - \beta_1 x_a)\bigr) &
    -{\delta_2}(\tfrac{\alpha_2}{\beta_1} - \tfrac{\beta_1}{\alpha_2})(\alpha_2^2 - \beta_2^2)
    \end{array}\!\!\right)\!.
\end{gather}
The parameter $\delta$ in \eqref{Lentriesex} is identified with the parameter $\delta_1$ of $C3_{(\delta_1;\,\delta_2;\,\delta_3)}$, while the above matrices also depend explicitly on the two other parameters $\delta_2,\delta_3$, of this equation.

The determinant of \eqref{Lentriesex} with \eqref{a3mataC}--\eqref{a3matbC} is given by
\begin{align}
\begin{split}
\det(\L)=D_\L^2
\bigl((\beta_1 x-\alpha_2 x_a)(\beta_1 x_a-\alpha_2 x)-\delta_2\tfrac{\alpha_2\beta_1}{2}(\tfrac{\alpha_2}{\beta_1}-\tfrac{\beta_1}{\alpha_2})^2\bigr) \phantom{.} \\
\times
\bigl((\beta_2 x-\alpha_2 x_b)(\beta_2 x_b-\alpha_2x)-\delta_2\tfrac{\alpha_2\beta_2}{2}(\tfrac{\alpha_2}{\beta_2}-\tfrac{\beta_2}{\alpha_2})^2\bigr).
\end{split}
\end{align}

\begin{prop}\label{prop:a30C}

For the three cases $(\delta_1,\delta_2,\delta_3)=(0,0,0),(1,0,0),(\tfrac{1}{2},0,\tfrac{1}{2})$, the Lax matrix \eqref{Lentriesex} defined with \eqref{a3mataC}--\eqref{a3matbC}, and the normalisation chosen as either
\begin{align}\label{a3norm1C}
    D_\L=\bigl((\beta_1 x - \alpha_2 x_a)(\beta_2 x_b - \alpha_2 x)\bigr)^{-1},\quad\mbox{or}
\\
\label{a3norm2C}
    D_\L=\bigl((\beta_1 x_a - \alpha_2 x)(\beta_2 x - \alpha_2 x_b)\bigr)^{-1},\phantom{\quad\mbox{or}}
\end{align}
satisfies the Lax equation \eqref{laxcomp} on solutions of $A3_{(\delta)}(\ccya;\ccf,\ccy,\cca,\ccb;\al,\gm)=0$.

\end{prop}

\begin{proof}

Using the definitions \eqref{laxmatA}, the Lax equation of Proposition \ref{prop:a30C} for the first normalisation \eqref{a3norm1C} may be written as
\begin{align}
    \Lb\M-\Mb\LL=\frac{(\alpha_1\alpha_2\gamma_1\gamma_2) A3_{(\delta)}(\ccya;\ccf,\ccy,\cca,\ccb;\al,\gm) \left(\!\begin{array}{cc} -\ccya & \delta_1\tfrac{\beta_1}{\gamma_2} + \delta_3\tfrac{\gamma_2}{\beta_1}\ccya^2 \\ 0 & 0 \end{array}\!\right)}
    {(\alpha_1\ccya - \gamma_1\cca)(\alpha_2 \ccya - \gamma_1 \ccf)(\gamma_2 \ccya - \alpha_2 \ccy)(\gamma_2 \ccya - \alpha_1 \ccb)},
\end{align}
while the Lax equation for the second normalisation \eqref{a3norm2C} may be written as
\begin{align}
    \Lb\M-\Mb\LL=\frac{(\alpha_1\alpha_2\gamma_1\gamma_2) A3_{(\delta)}(\ccya;\ccf,\ccy,\cca,\ccb;\al,\gm) \left(\!\begin{array}{cc} 0 & \delta_1\tfrac{\beta_1}{\gamma_1} + \delta_3\tfrac{\gamma_1}{\beta_1}\ccya^2 \\ 0 & \ccya \end{array}\!\right)}
    {(\alpha_1\ccya - \gamma_2\ccb)(\alpha_2 \ccya - \gamma_2 \ccy)(\gamma_1 \ccya - \alpha_2 \ccf)(\gamma_1 \ccya - \alpha_1 \cca)}.
\end{align}

\end{proof}

\subsubsection{\texorpdfstring{$A2_{(\delta_1;\,\delta_2)}$}{A2(delta1;delta2)}}


For convenience, the difference variable
\begin{align}\label{thtdef1}
    \theta_{ij}=\theta_i-\theta_j, \qquad i,j\in\{1,2,3,4\},
\end{align}
will be used in the following, where $\theta_i$, $i\in\{1,2,3,4\}$, represents one of the four components $\alpha_1,\alpha_2,\beta_1,\beta_2$ of $\al$ and $\bt$, as
\begin{align}\label{thtdef2}
    (\theta_1,\theta_2,\theta_3,\theta_4)=(\alpha_1,\alpha_2,\beta_1,\beta_2).
\end{align}

The type-A equation $A2_{(\delta_1;\,\delta_2)}$ is given by \eqref{a2dd}.  For the equation \eqref{a2dd}, the Lax matrix \eqref{Lentries} may be written in the form \eqref{Lentriesex}, where
\begin{gather}\label{a2mata}
    \Lxa=\left(\!\!\begin{array}{cc}
    \theta_{31} & \theta_{23}x_a -\theta_{24}x_b \\ 0 & \theta_{41}
    \end{array}\!\!\right)\!,
\quad
    \Lxb=\left(\!\!\begin{array}{cc}
    x_a \theta_{12} + x_b(\theta_{13}+\theta_{24}) & \theta_{34} x_a x_b \\ \theta_{34} & x_b \theta_{12} + x_a(\theta_{13}+\theta_{24})
    \end{array}\!\!\right)\!,
\\[0.3cm]
    \Lxc=\left(\!\!\begin{array}{cc}
    x_a x_b \theta_{41} & 0 \\ x_a \theta_{42} + x_b \theta_{23} & x_a x_b \theta_{31}
    \end{array}\!\!\right)\!,
\qquad 
    \Dxa=\left(\!\!\begin{array}{cc}
    0 & {\delta_2} \theta_{12} \theta_{34}(\theta_{13} + \theta_{24}) \\ 0 & 0
    \end{array}\!\!\right)\!,
\end{gather}
and entries of $\Dxb$ given by
\begin{align}
\begin{split}
    (\Dxb)_{11}=&\, {\delta_2}\bigl(2 \theta_{13}\theta_{24}^2 + \theta_{12}\theta_{34}(\theta_{12}-\theta_{34})\bigr), \\
    (\Dxb)_{12}=&\,
    \theta_{12} \theta_{34}(\theta_{13} + \theta_{24})(x_a+x_b-\theta_{13}^2-\theta_{24}^2-\theta_{12}\theta_{34})^{\delta_2} -2{\delta_2}\theta_{13} \theta_{14}(x_a \theta_{24} + x_b \theta_{32}), \\
    (\Dxb)_{21}=&\, 0, \\
    (\Dxb)_{22}=&\, {\delta_2}\bigl(2 \theta_{14}\theta_{23}^2-\theta_{12}\theta_{34}(\theta_{12}+\theta_{34})\bigr),
\end{split}
\end{align}
and entries of $\Dxc$ given by
\begin{align}\label{a2matb}
\begin{split}
    (\Dxc)_{11}=&\,
    \theta_{14} \theta_{23} \theta_{24}(x_b - \theta_{12} \theta_{34} - \theta_{24}^2)^{\delta_2} - {\delta_2} \theta_{14}(x_a \theta_{12} \theta_{24} - x_b \theta_{23} \theta_{13} ), \\
    (\Dxc)_{12}=&\, \theta_{13} \theta_{14} \Bigl(x_b \theta_{23}(\theta_{12} \theta_{43}- \theta_{13}^2)^{\delta_2} - x_a \theta_{24}(\theta_{12} \theta_{34} - \theta_{14}^2)^{\delta_2} \\
    &\phantom{\theta_{13} \theta_{14} \Bigl(} - \delta_2 \theta_{34} \bigl(x_a x_b-(\theta_{13}+\theta_{24})\theta_{12}\theta_{23}\theta_{24}\bigr)\!\Bigr), \\ 
    (\Dxc)_{21}=&\, {\delta_2} \theta_{23} \theta_{24} \theta_{43}, \\
    (\Dxc)_{22}=&\, \theta_{13} \theta_{23} \theta_{24}(x_a + \theta_{12} \theta_{34} - \theta_{23}^2)^{\delta_2} + {\delta_2} \theta_{13}\bigl(x_a \theta_{14} \theta_{24} - x_b \theta_{12} \theta_{23}\bigr).
    \end{split}
\end{align}
The parameter $\delta$ in \eqref{Lentriesex} is identified with the parameter $\delta_1$ of $A2_{(\delta_1;\,\delta_2)}$, while the above matrices also depend explicitly on the other parameter $\delta_2$ of this equation.

The determinant of \eqref{Lentriesex} with \eqref{a2mata}--\eqref{a2matb} is given by
\begin{align}
\begin{split}
\det(\L)=D_\L^2\theta_{13}\theta_{14}
\Bigl((x-x_a)^2-\delta_1\theta_{23}^2 \bigl(2(x+x_a)-\theta_{23}^2\bigr)^{\delta_2}\Bigr)\phantom{.} \\
\times\Bigl((x-x_b)^2-\delta_1\theta_{24}^2 \bigl(2(x+x_b)-\theta_{24}^2\bigr)^{\delta_2}\Bigr).
\end{split}
\end{align}
Note that for $\delta_2=0$, two of the above factors may be written as
\begin{align}\label{factorA2}
\begin{split}
    \bigl((x-x_a)^2-\delta_1\theta_{23}^2\bigr)&\bigl((x-x_b)^2-\delta_1\theta_{24}^2\bigr) =\prod_{\varepsilon\in\{-1,1\}}(x-x_a+\varepsilon\delta_1\theta_{23})(x-x_b+\varepsilon\delta_1\theta_{24}).
\end{split}
\end{align}
Two choices of the four factors on the right hand side will be used in the normalisation factor $D_\L$ to give compatible Lax matrices as follows.

\begin{prop}\label{prop:a20}

For the two cases $(\delta_1,\delta_2)=(0,0),(1,0)$,  the Lax matrix \eqref{Lentriesex} defined with \eqref{a2mata}--\eqref{a2matb}, and the normalisation
\begin{align}
    D_\L=\bigl((\alpha_1 + \tfrac{\varepsilon_1-1}{2}\beta_1  - \tfrac{\varepsilon_1+1}{2} \beta_2 )(x_a-x- \varepsilon_2 \delta_1 \theta_{23})(x_b-x+ \varepsilon_2 \delta_1 \theta_{24})\bigr)^{-1},\qquad \varepsilon_1,\varepsilon_2=\pm1,
\end{align}
satisfies the Lax equation \eqref{laxcomp} on solutions of $A2_{(\delta_1;\,\delta_2)}(\ccya;\ccf,\ccy,\cca,\ccb;\al,\gm)=0$.  
\end{prop}

\begin{proof}

Using the definitions \eqref{laxmatA}, the Lax equation of Proposition \ref{prop:a20} may be written as
\begin{align}
\begin{split}
    \Lb\M-\Mb\LL=&\frac{A2_{(\delta_1;\,\delta_2)}(\ccya;\ccf,\ccy,\cca,\ccb;\al,\gm)}
    {(\beta_1 - \tfrac{\varepsilon_1+1}{2}\gamma_1 + \tfrac{\varepsilon_1-1}{2} \gamma_2) (\ccya-\ccy+\delta_1 \varepsilon_2(\alpha_2-\gamma_2)) (\ccya-\ccb+\delta_1 \varepsilon_2(\alpha_1-\gamma_2))} \\[0.2cm]
    & \times
    \frac{-\left(\!\!\begin{array}{c} \delta_1 \varepsilon_2(\gamma_1-\beta_1) + \ccya \\ 1 \end{array}\!\!\right)\otimes
    \left(\!\!\begin{array}{c} 1 \\ \delta_1 \varepsilon_2 (\gamma_2-\beta_1) - \ccya  \end{array}\!\!\right)}
    {(\ccya-\ccf+\delta_1 \varepsilon_2(\gamma_1-\alpha_2)) (\ccya-\cca+\delta_1 \varepsilon_2(\gamma_1-\alpha_1))}.
\end{split}
\end{align}

\end{proof}


The type-C equation for this case is $C2_{(\delta_1;\,\delta_2;\,\delta_3)}$, given in \eqref{c2ddd}.  For \eqref{c2ddd}, the Lax matrix \eqref{Lentries} may be written in the form \eqref{Lentriesex}, where
\begin{gather}\label{a2mataC}
    \Lxa=-\left(\!\!\begin{array}{cc}
    1 & \theta_{43} \\ 0 & 1
    \end{array}\!\!\right)\!,
\quad
    \Lxb=\left(\!\!\begin{array}{cc}
    x_a + x_b & (\theta_{23}+\theta_{24})(x_b - x_a) \\ 0 & x_a + x_b
    \end{array}\!\!\right)\!,
\quad
    \Lxc=-x_ax_b\left(\!\!\begin{array}{cc}
    1 & \theta_{34} \\ 0 & 1
    \end{array}\!\!\right)\!,
\\[0.2cm]
    \Dxa=\left(\!\!\begin{array}{cc}
    0 & {\delta_3}\bigl(\theta_{34}(\theta_{13}+\theta_{14}-1) + 2(x_a \theta_{23} - x_b \theta_{24}) \bigr) \\ 0 & 0
    \end{array}\!\!\right)\!,
\end{gather}
and the components of $\Dxb$ are given by
\begin{align}
    \begin{split}
(\Dxb)_{11}=&\, \theta_{34}(2 \theta_{13})^{\delta_2} + {\delta_2}(\theta_{23}^2 + \theta_{24}^2), \\
(\Dxb)_{12}=&\, \theta_{43}\bigl(\theta_{13}^{1+\delta_2+\delta_3}+\theta_{14}^{1+\delta_2+\delta_3} + 2 {\delta_3}(\theta_{23}\theta_{24} - x_a x_b)\bigr) \\
& - {\delta_3}(\theta_{13}+\theta_{14} -1)(\theta_{23}+\theta_{24})(x_a-x_b), \\
(\Dxb)_{21}=&\, 2 {\delta_2} \theta_{34}, \\
(\Dxb)_{22}=&\, \theta_{43}(2 \theta_{14})^{\delta_2} + {\delta_2}(\theta_{23}^2+\theta_{24}^2),
    \end{split}
\end{align}
and the components of $\Dxc$ are given by
\begin{align}\label{a2matbC}
    \begin{split}
(\Dxc)_{11}=&\, \theta_{23}\theta_{24}(x_b-2\theta_{14}\theta_{34} - \theta_{23}\theta_{24})^{\delta_2} - x_a \theta_{24}(\theta_{12}+\theta_{14})^{\delta_2} + x_b \theta_{23}(\theta_{13}+\theta_{14})^{\delta_2}, \\
(\Dxc)_{12}=&\, \theta_{23}\theta_{24}\theta_{34}(2\theta_{13}\theta_{14} - \theta_{23}\theta_{24} + x_a + x_b)^{\delta_2}(\theta_{13}+\theta_{14})^{\delta_3}  \\
& + (\theta_{13}^{1+\delta_2+\delta_3}+\theta_{14}^{1+\delta_2+\delta_3})(x_a \theta_{24} - x_b \theta_{23}) + {\delta_3}\theta_{34}(1 - \theta_{13}-\theta_{14})x_a x_b, \\
(\Dxc)_{21}=&\, 2 {\delta_2} (\theta_{42} x_a + \theta_{23} x_b - \theta_{23}\theta_{24}\theta_{34}), \\
(\Dxc)_{22}=&\, \theta_{23}\theta_{24}(x_a + 2\theta_{13}\theta_{34}-\theta_{23}\theta_{24})^{\delta_2}  + x_a \theta_{24}( \theta_{13} + \theta_{14})^{\delta_2} - x_b \theta_{23}(\theta_{12}+\theta_{13})^{\delta_2}.
    \end{split}
\end{align}
The parameter $\delta$ in \eqref{Lentriesex} is identified with the parameter $\delta_1$ of $C2_{(\delta_1;\,\delta_2;\,\delta_3)}$, while the above matrices also depend explicitly on the two other parameters $\delta_2,\delta_3$, of this equation.

The determinant of \eqref{Lentriesex} with \eqref{a2mataC}--\eqref{a2matbC} is given by
\begin{align}
\begin{split}
\det(\L)=D_\L^2\Bigl((x-x_a)^2-\delta_1\theta_{23}^2\bigl(2(x+x_a)-\theta_{23}^2\bigr)^{\delta_2}\Bigr)
\Bigl((x-x_b)^2-\delta_1 \theta_{24}^2\bigl(2(x+x_b)-\theta_{24}^2\bigr)^{\delta_2}\Bigr).
\end{split}
\end{align}
For $\delta_2=0$, this factorises as in \eqref{factorA2}.

\begin{prop}\label{prop:a20C}

For the three cases $(\delta_1,\delta_2,\delta_3)=(0,0,0),(1,0,0),(1,0,1)$, the Lax matrix \eqref{Lentriesex} defined with \eqref{a2mataC}--\eqref{a2matbC}, and the normalisation
\begin{align}
    D_\L=\Bigl(\bigl(x-x_a-\varepsilon\delta_1\theta_{23}\bigr)\bigl(x-x_b+\varepsilon\delta_1\theta_{24}\bigr)\Bigr)^{-1},\qquad \varepsilon=\pm1,
\end{align}
satisfies the Lax equation \eqref{laxcomp} on solutions of $A2_{(\delta_1;\,\delta_2)}(\ccya;\ccf,\ccy,\cca,\ccb;\al,\gm)=0$.
\end{prop}

\begin{proof}

Using the definitions \eqref{laxmatA}, the Lax equation of Proposition \ref{prop:a20C} may be written as
\begin{align}
\begin{split}
    \Lb\M-\Mb\LL=&\frac{(2) A2_{(\delta_1;\,\delta_2)}(\ccya;\ccf,\ccy,\cca,\ccb;\al,\gm)}
    {\bigl(\ccya-\ccy+(\gamma_2 - \alpha_2)\delta_1\varepsilon\bigr)\bigl(\ccya-\ccb+(\gamma_2-\alpha_1)\delta_1\varepsilon\bigr)} \\[0.2cm]
    &\times\frac{\left(\!\begin{array}{cc} 
    -\delta_1\tfrac{1+\varepsilon}{2} & \bigl((\beta_1  + \tfrac{\varepsilon-1}{2}\gamma_1 - \tfrac{\varepsilon+1}{2}\gamma_2)\delta_1-\ccya\bigr)^{1+\delta_3} \\[0.15cm] 0 & \delta_1\tfrac{1-\varepsilon}{2}
    \end{array}\!\right)}
    {\bigl(\ccya-\ccf+(\alpha_2 - \gamma_1)\delta_1\varepsilon\bigr) \bigl(\ccya-\cca+(\alpha_1 - \gamma_1)\delta_1\varepsilon\bigr)}\!.
\end{split}
\end{align}

\end{proof}

Finally, for $(\delta_1,\delta_2)=(0,0)$, the type-A equation $A2_{(\delta_1;\,\delta_2)}$ also satisfies CAFCC together with the type-C equation given by \eqref{c1}.  Then using \eqref{c1} results in another Lax matrix \eqref{Lentries} for $A2_{(\delta_1=0;\,\delta_2=0)}$.

\begin{prop}\label{prop:C1}

The matrix
\begin{align}
\begin{split}
   \L=\frac{\left(\!\!\begin{array}{cc} (x-x_a)(x-x_b) & 2 \bigl(\beta_2 (x - x_a) + \beta_1 (x - x_b) + \alpha_2 (x_a+x_b-2 x)\bigr) \\ 0 & -(x-x_a)(x-x_b) \end{array}\!\!\right)}
   {(x-x_a)(x-x_b)}\!,
\end{split}
\end{align}
satisfies the Lax equation \eqref{laxcomp} on solutions of $A2_{(\delta_1=0;\,\delta_2=0)}(\ccya;\ccf,\ccy,\cca,\ccb;\al,\bt)$.

\end{prop}

\begin{proof}

Using the definitions \eqref{laxmatA}, the Lax equation of Proposition \ref{prop:C1}, may be written as
\begin{align}
    \Lb\M-\Mb\LL=\frac{-2A2_{(\delta_1=0;\,\delta_2=0)}(\ccya;\ccf,\ccy,\cca,\ccb;\al,\bt)}{(\ccya-\ccy)(\ccya-\ccb)(\ccya-\ccf)(\ccya-\cca)}
    \left(\!\!\begin{array}{cc} 0 & 1 \\ 0 & 0\end{array}\!\!\right)\!.
\end{align}

\end{proof}

\subsection{Lax matrices for type-B equations (approach 2)}\label{sec:LaxB}

Following the approach of Section \ref{sec:laxmethod2}, this section will give Lax matrices \eqref{Lentries2} which are constructed only from type-C CAFCC equations, where the compatibility condition \eqref{laxcomp2} this time is satisfied on solutions of type-B equations.  Similarly to \eqref{Lentriesex}, the Lax matrix \eqref{Lentries2} will be given here in the form
\begin{align}\label{Lentriesex2}
    L(x;x_a,x_c;\al,\bt)=D_\L\bigl(x^2\Lxa+x\Lxb+\Lxc+\delta_1(x^2\Dxa+x\Dxb+\Dxc)\bigr),
\end{align}
where $\delta_1$ is a parameter appearing for type-C equations, and each of the $\Lxa(x_a,x_c;\al,\bt)$, $\Lxb(x_a,x_c;\al,\bt)$, $\Lxc(x_a,x_c;\al,\bt)$, $\Dxa(x_a,x_c;\al,\bt)$, $\Dxb(x_a,x_c;\al,\bt)$, $\Dxc(x_a,x_c;\al,\bt)$, are $2\times2$ matrices.  The Lax matrix \eqref{Lentriesex2} will be given by specifying each of the $\Lxa$, $\Lxb$, $\Lxc$, $\Dxa$, $\Dxb$, $\Dxc$, as well as a valid normalisation factor $D_\L$.

\subsubsection{\texorpdfstring{$B3_{(\delta_1;\,\delta_2;\,\delta_3)}$}{B3(delta1;delta2;delta3)}}

On the face-centered cube, a set of CAFCC equations in the form \eqref{CAFCCeqsBC} is collectively given by the type-A equation $A3_{(\delta)}$ in \eqref{a3d}, the type-B equation $B3_{(\delta_1;\,\delta_2;\,\delta_3)}$ in \eqref{b3ddd}, and the type-C equation $C3_{(\delta_1;\,\delta_2;\,\delta_3)}$ in \eqref{c3ddd}.  For the type-B and type-C equations, the values of $(\delta_1,\delta_2,\delta_3)$ can be $(0,0,0)$, $(1,0,0)$, $(\tfrac{1}{2},0,\tfrac{1}{2})$, $(\tfrac{1}{2},\tfrac{1}{2},0)$, and then for the type-A equation $\delta=2\delta_2$.

By the method of Section \ref{sec:laxmethod2}, the type-C equation \eqref{c3ddd} can be used to obtain a Lax matrix for the type-B equation \eqref{b3ddd}.  For equation \eqref{c3ddd}, the Lax matrix \eqref{Lentries2} may be written in the form \eqref{Lentriesex2}, where
\begin{gather}
\label{b3mata}
    \Lxa= \alpha_2\left(\!\!\begin{array}{cc}
    -\beta_1 & \beta_2 x_c \\ 0 & 0 
    \end{array}\!\!\right)\!,
\quad
    \Lxb=-\left(\!\!\begin{array}{cc}
    -\alpha_2^2 x_a & \beta_1 \beta_2 x_a x_c \\[0.1cm] \beta_1 \beta_2 & -\alpha_2^2 x_c
    \end{array}\!\!\right)\!,
\quad
    \Lxc=\alpha_2 x_a\left(\!\!\begin{array}{cc}
    0 & 0 \\[0.1cm] \beta_2 & -\beta_1 x_c
    \end{array}\!\!\right)\!,
\\[0.0cm]
    \Dxa=\frac{2 {\delta_3}}{\alpha_1}\left(\!\!\begin{array}{cc}
    0 & (\beta_1^2 - \alpha_2^2)\beta_2 x_a \\[0.1cm] 0 & (\beta_2^2 - \alpha_2^2)\beta_1
    \end{array}\!\!\right)\!,
\quad
    \Dxb=\frac{\alpha_2(\beta_1^2-\beta_2^2)}{\alpha_1}\left(\!\!\begin{array}{cc}
    2 {\delta_2} x_c & \tfrac{\alpha_1^2}{\beta_1\beta_2} \\ 0 & 2 {\delta_3}x_a
    \end{array}\!\!\right)\!,
\\[0.3cm]
\label{b3matb}
    \Dxc=\left(\!\!\begin{array}{cc}
    -{\delta_2}\tfrac{(\alpha_2^2 - \beta_2^2)(\alpha_1(\alpha_2^2 - \beta_1^2) + 2 \alpha_2 \beta_1^2 x_a x_c)}{\alpha_1 \alpha_2 \beta_1} & 
    (\beta_2-\tfrac{\alpha_2^2}{\beta_2})\bigl(\alpha_1 x_a - {\delta_2}(\alpha_2 - \tfrac{\beta_1^2}{\alpha_2})x_c\bigr) \\
    2 {\delta_2}(\beta_1^2 - \alpha_2^2)x_c \tfrac{\beta_2}{\alpha_1} & \alpha_1(\beta_1-\tfrac{\alpha_2^2}{\beta_1})
    \end{array}\!\!\right)\!.
\end{gather}

Note that the above matrices have an additional dependence on the parameters $\delta_2,\delta_3$ of \eqref{c3ddd}.

The determinant of \eqref{Lentriesex2} with \eqref{b3mata}--\eqref{b3matb}, is given by
\begin{align}\label{detb3}
\begin{split}
\det(\L)=D_\L^2(\alpha_2^2-\beta_2^2)\bigl((\alpha_2 x - \beta_1 x_a)(\alpha_2 x_a - \beta_1 x) - \delta_2 \tfrac{\alpha_2 \beta_1}{2}(\tfrac{\alpha_2}{\beta_1}-\tfrac{\beta_1}{\alpha_2})^2\bigr)\phantom{.} \\
\times\bigl(x x_c - \delta_1 \tfrac{\alpha_1}{\beta_1} - \delta_2 \tfrac{\beta_1}{\alpha_1} x_c^2 - \delta_3 \tfrac{\beta_1}{\alpha_1} x^2\bigr).
\end{split}
\end{align}

For a variable $x_i$, let ${x_i}^-$ and ${x_i}^+$ respectively denote
\begin{align}
    {x_i}^-=x_i-(x_i^2-1)^{\frac{1}{2}},\qquad {x_i}^+=x_i+(x_i^2-1)^{\frac{1}{2}}.
\end{align}
To obtain a valid normalisation factor $D_\L$, note that for $(\delta_1,\delta_2,\delta_3)=(\frac{1}{2},0,\frac{1}{2})$, the last factor in the determinant \eqref{detb3} may be written as
\begin{align}
    -2\alpha_1\beta_1(x x_c-\tfrac{1}{2}\tfrac{\alpha_1}{\beta_1}-\tfrac{1}{2}\tfrac{\beta_1}{\alpha_1}x^2)=
    \bigl(\beta_1 x - \alpha_1x_c^-\bigr) \bigl(\beta_1 x - \alpha_1x_c^+\bigr).
\end{align}
One of the two factors on the right hand side will be used in the normalisation factor $D_\L$ to give compatible Lax matrices as follows.

\begin{prop}\label{prop:b30}

For the three cases $(\delta_1,\delta_2,\delta_3)=(0,0,0),(1,0,0),(\frac{1}{2},0,\frac{1}{2})$, the Lax matrix \eqref{Lentriesex2} defined with \eqref{b3mata}--\eqref{b3matb}, and the normalisation
\begin{align}\label{normb3}
    D_\L=(\beta_1 x - \alpha_1\overline{x}_c)^{-2\delta_3},\qquad \overline{x}_c=x_c^+\mbox{ or }x_c^-,
\end{align}
satisfies the Lax equation \eqref{laxcomp2} on solutions of $B3_{(\delta_1;\,\delta_2;\,\delta_3)}(\ccyc;\ccf,\cce,\ccy,\ccd;\gm,\bt)=0$.
\end{prop}

\begin{proof}

First consider the case of $\delta_3=0$, for which the normalisation in \eqref{normb3} is $D_\L=1$.  Then using the definitions \eqref{laxmatb}, the Lax equation of Proposition \ref{prop:b30} for the two cases $(\delta_1,\delta_2,\delta_3)=(0,0,0),(1,0,0)$, may be written as
\begin{align}
    \Lb\M-\Mb\LL=\frac{B3_{(\delta_1;\,\delta_2;\,\delta_3)}(\ccyc;\ccf,\cce,\ccy,\ccd;\gm,\bt)}{\ccyc(\alpha_2^2 - \gamma_1^2)(\ccf - \delta_1\tfrac{\beta_1}{\gamma_2\ccyc})(\cce - \delta_1\tfrac{\beta_2}{\gamma_2\ccyc})}
    \left(\!\begin{array}{cc} \gamma_2 \ccyc \\ \alpha_2 \end{array}\!\right)\otimes
    \left(\!\begin{array}{cc} -\gamma_1 \\ \alpha_2\ccyc \end{array}\!\right)\!.
\end{align}

Similarly, using the definitions \eqref{laxmatb} the Lax equation of Proposition \ref{prop:b30} for the case $(\delta_1,\delta_2,\delta_3)=(\frac{1}{2},0,\frac{1}{2})$ may be written as
\begin{align}\label{dum1}
\begin{split}
    \Lb\M&-\Mb\LL= \\
    &\frac{\beta_1 \beta_2 \alpha_2 \gamma_2 \ccyc B3_{(\delta_1;\,\delta_2;\,\delta_3)}(\ccyc;\ccf,\cce,\ccy,\ccd;\gm,\bt) E}
    {(\alpha_2^2 - \gamma_1^2)(\gamma_1 \ccy - \beta_1 \overline{\ccyc})(\gamma_1 \ccd - \beta_2 \overline{\ccyc})(\beta_1^2 + \gamma_2^2 \ccf^2 - 2 \beta_1 \gamma_2 \ccf \ccyc)(\beta_2^2 + \gamma_2^2 \cce^2 - 2 \beta_2 \gamma_2 \cce \ccyc)},
\end{split}
\end{align}
where the entries of the matrix $E$ are given by
\begin{align}\label{dum2}
    \hspace{-0.13cm}\begin{split}
E_{11}=&\, \tfrac{2 \gamma_1}{\alpha_2}\Bigl(\! \bigl(\znb(\alpha_2^2-\gamma_2^2) -2 \alpha_2^2 \ccyc\bigr) (\ccf \cce \gamma_2^2 - \beta_1 \beta_2) + \bigl(\gamma_2(\beta_2 \ccf+ \beta_1 \cce) - 2 \beta_1 \beta_2 \ccyc\bigr)(\gamma_2^2 \znb^2 + \alpha_2^2)\! \Bigr), \\
E_{12}=&\, (\alpha_2^2+\tfrac{\gamma_1^2 \gamma_2^2}{\alpha_2^2})(\gamma_2 \ccf -\beta_1 \znb)(\gamma_2 \cce-\beta_2 \znb) + \gamma_1^2 \bigl((\beta_1 - 2 \gamma_2 \ccf \ccyc)(\beta_2 - 2 \gamma_2 \cce \ccyc) -2 \ccf \cce \gamma_2^2\bigr) \\
& + \gamma_2 \znb \bigl( (\gamma_1^2 - \gamma_2^2 \znb^2)(\beta_2 \ccf+ \beta_1 \cce) + \gamma_2(\gamma_2^2-\gamma_1^2)\ccf \cce \znb + \beta_1 \beta_2 \gamma_2 \znb^3\bigr), \\
E_{21}=&\, 4 \gamma_1 \gamma_2\bigl(\beta_1 \beta_2 + \gamma_2 \ccf(\beta_2 \znb-\gamma_2 \cce) + \beta_1 \znb (\gamma_2 \cce - 2 \beta_2 \ccyc)\bigr), \\
E_{22}=&\, \tfrac{2\gamma_2}{ \alpha_2} \Bigl(\! \bigl(\znb(\alpha_2^2-\gamma_1^2) + 2 \gamma_1^2 \ccyc\bigr)(\ccf \cce \gamma_2^2 - \beta_1 \beta_2) - \bigl(\gamma_2(\beta_2 \ccf +\beta_1 \cce) - 2 \beta_1 \beta_2 \ccyc\bigr)(\alpha_2^2 \znb^2 + \gamma_1^2)\! \Bigr),
    \end{split}
\end{align}
where $\znb=z_n^\pm$, according to the choice of $\overline{x}_c=x_c^\pm$ in \eqref{normb3}.

\end{proof}


\subsubsection{\texorpdfstring{$B2_{(\delta_1;\,\delta_2;\,\delta_3)}$}{B2(delta1;delta2;delta3)}}

For convenience, the definitions \eqref{thtdef1}, \eqref{thtdef2}, will be used in the following.

On the face-centered cube, a set of CAFCC equations in the form \eqref{CAFCCeqsBC} is collectively given by the type-A equation $A2_{(\delta_1;\,\delta_2)}$ in \eqref{a2dd}, the type-B equation $B2_{(\delta_1;\,\delta_2;\,\delta_3)}$ in \eqref{b2ddd}, and the type-C equation $C2_{(\delta_1;\,\delta_2;\,\delta_3)}$ in \eqref{c2ddd}. The values of $(\delta_1,\delta_2,\delta_3)$ can be $(0,0,0)$, $(1,0,0)$, $(1,0,1)$, $(1,1,0)$.

For the type-C equation \eqref{c2ddd}, the Lax matrix \eqref{Lentries2}, may be written in the form \eqref{Lentriesex2}, where
\begin{gather}\label{b2mata}
    \Lxa=\left(\!\!\begin{array}{cc}
    1 & \theta_{34}-x_c \\ 0 & 0 
    \end{array}\!\!\right)\!,
\;\;
    \Lxb=-\left(\!\!\begin{array}{cc}
    x_a & x_a (\theta_{23} + \theta_{24} - x_c) \\ -1 & \theta_{23} + \theta_{24} + x_c
    \end{array}\!\!\right)\!,
\;\;
    \Lxc=\left(\!\!\begin{array}{cc}
    0 & 0 \\ -x_a & x_a (x_c + \theta_{34})
    \end{array}\!\!\right)\!,
\\[0.3cm]
    \Dxa=\left(\!\!\begin{array}{cc}
    -2({\delta_2}+{\delta_3}) & 2 {\delta_2}(x_c-\theta_{34}) + {\delta_3}\bigl(2(x_c - x_a \theta_{23}) -\theta_{34}(\theta_{13}+\theta_{14}+1)\bigr) \\ 0 & -2 {\delta_3} \theta_{24}
    \end{array}\!\!\right)\!,
\end{gather}
and the entries of $\Dxb$ are given by
\begin{align}
    \begin{split}
(\Dxb)_{11}=&\, \theta_{34}\bigl(2(x_c-\theta_{14})\bigr)^{\delta_2}(-1)^{\delta_3} + ({\delta_2}+{\delta_3})2 x_a + {\delta_2}(\theta_{23}^2+\theta_{24}^2), \\
(\Dxb)_{12}=&\, (\theta_{34} + {\delta_3} x_a)\bigl(x_c(\theta_{31}+\theta_{41})^{\delta_2}(-1)^{\delta_3} +2 {\delta_3} \theta_{12}^2 + (\theta_{31}^{1+{\delta_2}}+\theta_{41}^{1+{\delta_2}}) (\theta_{32}+\theta_{42})^{\delta_3}\bigr) \\
   &+ 2{\delta_2}\bigl(x_a(\theta_{24} - x_c) + \theta_{23}(x_a - x_c \theta_{24}) \bigr) -  {\delta_3}x_a\bigl(x_c + 2 \theta_{12}^2 - (\theta_{23} + \theta_{24})\bigr) , \\
(\Dxb)_{21}=&\, -2({\delta_2}+{\delta_3}), \\
(\Dxb)_{22}=&\, 2({\delta_2}+{\delta_3})x_c + 2{\delta_3} x_a \theta_{34} + (2 {\delta_2}+{\delta_3})(\theta_{23}+\theta_{24})(\theta_{13}+\theta_{14}+1)^{\delta_3},
    \end{split}
\end{align}
and the entries of $\Dxc$ are given by
\begin{align}\label{b2matb}
    \begin{split}
(\Dxc)_{11}=&\, \theta_{42}\Bigl(\theta_{23}\bigl(\theta_{43}(\theta_{12}+\theta_{13}-2 x_c) +\theta_{23}^2 -x_a\bigr)^{\delta_2} + x_a(2 x_c -\theta_{13}-\theta_{14})^{\delta_2}\Bigr)(-1)^{\delta_3}, \\
(\Dxc)_{12}=&\, \theta_{42}\biggl(\theta_{23}\bigl(\theta_{43}( \theta_{23} \theta_{24} -2 \theta_{13} \theta_{14})^{\delta_2}(\theta_{31}+\theta_{41})^{\delta_3} - x_c(2 \theta_{14}\theta_{34} + \theta_{23}\theta_{24})^{\delta_2} (-1)^{\delta_3}\bigr) \\[-0.15cm]
 & + x_a\Bigl((\theta_{31}+\theta_{41})^{1+{\delta_2}+{\delta_3}} +{\delta_2}\theta_{23}\theta_{34} - ({\delta_2} +{\delta_3}) 2 \theta_{13}\theta_{14} + x_c(\theta_{21}+\theta_{41})^{\delta_2} (-1)^{\delta_3}\Bigr)\! \biggr), \\[-0.1cm]
(\Dxc)_{21}=&\, \theta_{32}(2 x_c -\theta_{12}-\theta_{13})^{\delta_2}(-1)^{\delta_3} + ({\delta_2}+{\delta_3})2 x_a, \\
(\Dxc)_{22}=&\, \theta_{23}\bigl(-(\theta_{31}^{1+{\delta_2}+{\delta_3}}+\theta_{41}^{1+{\delta_2}+{\delta_3}})  - x_c(2 \theta_{41}-\theta_{23})^{\delta_2}(-1)^{\delta_3} + {\delta_2} \theta_{24}\theta_{34} \bigr) \\
& - (2 {\delta_2} + {\delta_3})(x_c+\theta_{34})x_a - {\delta_3}\bigl(x_c + \theta_{34}(\theta_{13}+\theta_{14})\bigr)x_a.
    \end{split}
\end{align}
Note that the above matrices have an additional dependence on the parameters $\delta_2,\delta_3$ of \eqref{c2ddd}.

The determinant of \eqref{Lentriesex2} with \eqref{b2mata}--\eqref{b2matb} is given by
\begin{align}
\begin{split}
\det(\L)=2 D_\L^2\theta_{24}\bigl(x(2 \theta_{13}-x)^{\delta_3}+\delta_1 x_c(2 \theta_{13}-x_c)^{\delta_2}-\delta_1 \theta_{13}^{1+\delta_2+\delta_3}\bigr)\phantom{.} \\
\times\bigl(\delta_1 \theta_{23}^2 (2(x+x_a)-\theta_{23}^2)^{\delta_2}-(x-x_a)^2\bigr).
\end{split}
\end{align}
To obtain a valid normalisation factor $D_\L$, note that for $(\delta_1,\delta_2,\delta_3)=(1,0,1)$, the last factor on the first line of the above determinant may be written as
\begin{align}
    -\bigl(x(2\theta_{13}-x)+x_c-\theta_{13}^2\bigr)=(x+\beta_1-\alpha_1+x_c^{\frac{1}{2}})(x+\beta_1-\alpha_1-x_c^{\frac{1}{2}}).
\end{align}
One of the two factors on the right hand side will be used in the normalisation factor $D_\L$, to give compatible Lax matrices as follows.

\begin{prop}\label{prop:b20}

For the three cases $(\delta_1,\delta_2,\delta_3)=(0,0,0),(1,0,0),(1,0,1)$, the Lax matrix \eqref{Lentriesex2} defined with \eqref{b2mata}--\eqref{b2matb}, and the normalisation
\begin{align}
    D_\L=\bigl(x+\beta_1-\alpha_1 + \varepsilon x_c^{\frac{1}{2}}\bigr)^{-\delta_3},\qquad\varepsilon=\pm1,
\end{align}
satisfies the Lax equation \eqref{laxcomp2} on solutions of $B2_{(\delta_1;\,\delta_2;\,\delta_3)}(\ccyc;\ccf,\cce,\ccy,\ccd;\gm,\bt)=0$.

\end{prop}

\begin{proof}

Using the definitions \eqref{laxmatb}, the Lax equation of Proposition \ref{prop:b20} may be written as
\begin{align}
\begin{split}
    \Lb\M-\Mb\LL=&\frac{B2_{(\delta_1;\,\delta_2;\,\delta_3)}(\ccyc;\ccf,\cce,\ccy,\ccd;\gm,\bt)}
    {\bigl(\ccf+\delta_1(\gamma_2-\beta_1+\znb^{\delta_3}\ccyc^{1-\delta_3})\bigr)\bigl(\cce+\delta_1(\gamma_2-\beta_2+\znb^{\delta_3}\ccyc^{1-\delta_3})\bigr)} \\[0.2cm]
    &\times\frac{\left(\!\!\begin{array}{c} 
    (\gamma_2-\alpha_2+\znb^{\delta_3}\ccyc^{1-\delta_3})^{1+\delta_3} \\ 1
    \end{array}\!\!\right)\otimes
    \left(\!\!\begin{array}{c} 
    1 \\ -(\alpha_2 - \gamma_1 +\znb^{\delta_3}\ccyc^{1-\delta_3})^{1+\delta_3} 
    \end{array}\!\!\right)}
    {2(\alpha_2-\gamma_1)(\ccy+\gamma_1-\beta_1 - \znb)^{\delta_3}(\ccd+\gamma_1-\beta_2 -  \znb)^{\delta_3}},
\end{split}
\end{align}
where $\znb=-\varepsilon\ccyc^{\frac{1}{2}}$.

\end{proof}

\begin{remark}
Note that for $(\delta_1,\delta_2,\delta_3)=(0,0,0)$, the type-B equations here and in the preceding case are equivalent to a CAC quad equation known as $D4$ \cite{BollThesis}.  Thus the equations $B3_{(\delta_1;\delta_2;\delta_3)}$, and $B2_{(\delta_1;\delta_2;\delta_3)}$, could be regarded as two different extensions of $D4$, which satisfy CAFCC rather than CAC.
\end{remark}

\subsubsection{\texorpdfstring{$D1$}{D1}}

On the face-centered cube, a set of CAFCC equations in the form \eqref{CAFCCeqsBC} is collectively given by the type-A equation $A2_{(\delta_1=0;\,\delta_2=0)}$ in \eqref{a2dd}, the type-B equation $D1$ in \eqref{d1}, and the type-C equation $C1$ in \eqref{c1}.  Using the method of Section \ref{sec:laxmethod2}, the type-C equation \eqref{c1} can be found to give a Lax matrix for the type-B equation \eqref{d1}, as follows.

\begin{prop}\label{prop:D1}

The matrix
\begin{align}
\begin{split}
   \L=\frac{\left(\!\!\begin{array}{cc} x(x-x_a) &  (x-x_a)x x_c - 2 \beta_2 x_a + 2 (\beta_1+\beta_2)x +2 \alpha_2(x_a-2x) \\ 
   x-x_a & (x-x_a)x_c +2(\beta_1-\alpha_2) \end{array}\!\!\right)}
   {(x-x_a)}\!,
\end{split}
\end{align}
satisfies the Lax equation \eqref{laxcomp2} on solutions of $D1(\ccf,\cce,\ccy,\ccd)=0$.

\end{prop}

\begin{proof}

Using the definitions \eqref{Lentries2}, the Lax equation for Proposition \ref{prop:D1} may be written as
\begin{align}
    \Lb\M-\Mb\LL=\frac{D1(\ccf,\cce,\ccy,\ccd)}{2(\gamma_1-\alpha_2)}
    \left(\!\!\begin{array}{cc} -\ccyc & \frac{8(\gamma_1-\gamma_2)^2}{(\ccy-\ccf)(\ccd-\cce)} -\ccyc^2 \\ 1 & \ccyc\end{array}\!\!\right)\!.
\end{align}

\end{proof}

\begin{remark}
Note that this gives a Lax matrix for $D1$, which is also equivalent to a regular quad equation which satisfies CAC \cite{BollThesis}.
\end{remark}


\begin{appendices}
\numberwithin{equation}{section}

\section{List of CAFCC equations}\label{app:equations}

\subsection{\texorpdfstring{$A3_{(\delta)}$}{A3(delta)}, \texorpdfstring{$B3_{(\delta_1;\,\delta_2;\,\delta_3)}$}{B3(delta1;delta2;delta3)}, \texorpdfstring{$C3_{(\delta_1;\,\delta_2;\,\delta_3)}$}{C3(delta1;delta2;delta3)}}\label{app:b3}

\begin{align}\label{a3d}
\begin{split}
A3&_{(\delta)}(x;x_a,x_b,x_c,x_d;\al,\bt)= \\ 
& x\bigl((\tfrac{\beta_1}{\beta_2}-\tfrac{\beta_2}{\beta_1})(x_ax_b-x_cx_d) + (\tfrac{\alpha_1}{\alpha_2}-\tfrac{\alpha_2}{\alpha_1})(x_ax_c-x_bx_d) - (\tfrac{\alpha_1\alpha_2}{\beta_1\beta_2}-\tfrac{\beta_1\beta_2}{\alpha_1\alpha_2})(x_ax_d-x_bx_c)\bigr)  \\
&+(\tfrac{\alpha_2}{\beta_1}-\tfrac{\beta_1}{\alpha_2})(x_a x^2 - x_bx_cx_d) - (\tfrac{\alpha_2}{\beta_2}-\tfrac{\beta_2}{\alpha_2})(x_b x^2 - x_ax_cx_d)   
 - (\tfrac{\alpha_1}{\beta_1}-\tfrac{\beta_1}{\alpha_1})(x_c x^2 - x_ax_bx_d)   \\
 & + (\tfrac{\alpha_1}{\beta_2}-\tfrac{\beta_2}{\alpha_1})(x_d x^2 - x_ax_bx_c) - \delta(\tfrac{\alpha_1}{\alpha_2}-\tfrac{\alpha_2}{\alpha_1})(\tfrac{\beta_1}{\beta_2}-\tfrac{\beta_2}{\beta_1})\bigl(\tfrac{\alpha_1\alpha_2}{\beta_1\beta_2}-\tfrac{\beta_1\beta_2}{\alpha_1\alpha_2}\bigr)x \\
& + \delta\Bigl((\tfrac{\alpha_1}{\beta_1}-\tfrac{\beta_1}{\alpha_1})(\tfrac{\alpha_2}{\beta_2}-\tfrac{\beta_2}{\alpha_2})\bigl((\tfrac{\alpha_1}{\beta_2}-\tfrac{\beta_2}{\alpha_1})x_a + (\tfrac{\alpha_2}{\beta_1}-\tfrac{\beta_1}{\alpha_2})x_d\bigr)
 \\
& \phantom{+\delta\Bigl(} - (\tfrac{\alpha_1}{\beta_2}-\tfrac{\beta_2}{\alpha_1})(\tfrac{\alpha_2}{\beta_1}-\tfrac{\beta_1}{\alpha_2})\bigl((\tfrac{\alpha_1}{\beta_1}-\tfrac{\beta_1}{\alpha_1})x_b + (\tfrac{\alpha_2}{\beta_2}-\tfrac{\beta_2}{\alpha_2})x_c\bigr) 
\Bigr) =0.
\end{split}
\end{align}

The above equation satisfies CAFCC in the form \eqref{CAFCCeqsA}, for the two values $\delta=0,1$.

\begin{align}\label{b3ddd}
    \begin{split}
B3&_{(\delta_1;\,\delta_2;\,\delta_3)}(x;x_a,x_b,x_c,x_d;\al,\bt)= \\ 
& x_b x_c-x_a x_d + \tfrac{\delta_2}{2}(\tfrac{\alpha_2}{\alpha_1}-\tfrac{\alpha_1}{\alpha_2})(\tfrac{\beta_1}{\beta_2}-\tfrac{\beta_2}{\beta_1}) + 
{\delta_2}(\tfrac{\alpha_1}{\beta_2} x_a - \tfrac{\alpha_1}{\beta_1} x_b - \tfrac{\alpha_2}{\beta_2} x_c + \tfrac{\alpha_2}{\beta_1} x_d)x \\
& +{\delta_1}(\tfrac{\beta_2}{\alpha_1} x_a - \tfrac{\beta_1}{\alpha_1} x_b - \tfrac{\beta_2}{\alpha_2} x_c + \tfrac{\beta_1}{\alpha_2} x_d)x^{-1}+ 
{\delta_3}\bigl(x_a x_b \alpha_2(\tfrac{x_d}{\beta_2}-\tfrac{x_c}{\beta_1})+x_c x_d \alpha_1(\tfrac{x_a}{\beta_1}-\tfrac{x_b}{\beta_2})\bigr)x^{-1}=0.
    \end{split}
\end{align}

\begin{align}\label{c3ddd}
    \begin{split}
&C3_{(\delta_1;\,\delta_2;\,\delta_3)}(x;x_a,x_b,x_c,x_d;\al,\bt)= \\ 
& \Bigl(\alpha_2 (\beta_1 x_d - \beta_2 x_c) -{\delta_3}\bigl(\alpha_2^2 (\beta_1 x_b-\beta_2 x_a) + \beta_1 \beta_2(\beta_1 x_a - \beta_2 x_b)\bigr)\alpha_1^{-1}\Bigr)x^2  \\
& +\Bigl(\alpha_2^2 (x_b x_c - x_a x_d) + \beta_1 \beta_2 (x_a x_c - x_b x_d) + \alpha_2(\tfrac{\beta_2}{\beta_1}-\tfrac{\beta_1}{\beta_2})\bigl({\delta_1} \alpha_1 - {\delta_3}\tfrac{\beta_1 \beta_2}{\alpha_1}x_a x_b  +  {\delta_2}\tfrac{\beta_1 \beta_2}{\alpha_1} x_c x_d\bigr)\Bigr)x \\
& +\alpha_2 x_a x_b (\beta_2 x_d-\beta_1 x_c) + {\delta_1}\alpha_1\bigl(\beta_1 x_b-\beta_2 x_a + \alpha_2^2 (\tfrac{x_a}{\beta_2} - \tfrac{x_b}{\beta_1})\bigr)  \\
& + {\delta_2}\Bigl(\tfrac{(\alpha_2^2 - \beta_1^2)(\alpha_2^2 - \beta_2^2)}{2\alpha_2\beta_1\beta_2}(\beta_2 x_d-\beta_1 x_c) + \tfrac{x_c x_d}{\alpha_1}\bigl(\beta_1 \beta_2 (\beta_1 x_b-\beta_2 x_a) + \alpha_2^2 (\beta_1 x_a - \beta_2 x_b)\bigr)\Bigr)=0.
    \end{split}
\end{align}

The equations \eqref{a3d}, \eqref{b3ddd}, \eqref{c3ddd}, collectively satisfy CAFCC in the form \eqref{CAFCCeqsBC}, for the four values of $(\delta_1,\delta_2,\delta_3)=(0,0,0)$, $(1,0,0)$, $(\tfrac{1}{2},0,\tfrac{1}{2}),(\tfrac{1}{2},\tfrac{1}{2},0)$, where the parameter for \eqref{a3d} is $\delta=2\delta_2$.

\subsection{\texorpdfstring{$A2_{(\delta_1;\,\delta_2)}$}{A2(delta1;delta2)}, \texorpdfstring{$B2_{(\delta_1;\,\delta_2;\,\delta_3)}$}{B2(delta1;delta2;delta3)}, \texorpdfstring{$C2_{(\delta_1;\,\delta_2;\,\delta_3)}$}{C2(delta1;delta2;delta3)}}\label{app:b2}

Recall the definitions \eqref{thtdef1} and \eqref{thtdef2} given by
\begin{align}
    \theta_{ij}=\theta_i-\theta_j, \qquad i,j\in\{1,2,3,4\},\qquad (\theta_1,\theta_2,\theta_3,\theta_4)=(\alpha_1,\alpha_2,\beta_1,\beta_2).
\end{align}
\begin{align}\label{a2dd}
\begin{split}
 A2&_{(\delta_1;\,\delta_2)}(x;x_a,x_b,x_c,x_d;\al,\bt)= \\ 
& \theta_{23}(x_ax^2 -x_bx_cx_d) - \theta_{24}(x_bx^2 -x_ax_cx_d) - \theta_{13}(x_cx^2 -x_ax_bx_d) + \theta_{14}(x_dx^2 -x_ax_bx_c)  \\
& + \Bigl(\theta_{34}(x_ax_b-x_cx_d) + \theta_{12}(x_ax_c-x_bx_d) - (\theta_{13}+\theta_{24})(x_ax_d-x_bx_c)  \Bigr)x \\
& + {\delta_1}\Bigl(\theta_{14}\theta_{23}(\theta_{13}x_b+\theta_{24}x_c)(2x-\theta_{12}\theta_{34})^{\delta_2} - \theta_{13}\theta_{24}(\theta_{13}x_a+\theta_{23}x_d)(2x+\theta_{12}\theta_{34})^{\delta_2}
   \Bigr) \\
& + {\delta_1}x\theta_{12}\theta_{34}(\theta_{13}+\theta_{24})\bigl(x+x_a+x_b+x_c+x_d - \theta_{12}^2-\theta_{13}\theta_{23}-\theta_{14}\theta_{24}\bigr)^{\delta_2}  \\
& + {\delta_2}\Bigl( x_a \theta_{13}\theta_{14}(\theta_{24}\theta_{14}^2-\theta_{34}x_b) - x_b\theta_{13}\theta_{23}(\theta_{14}\theta_{13}^2-\theta_{12}x_d) - x_c\theta_{14}\theta_{24}(\theta_{23}\theta_{24}^2+\theta_{12}x_a)     \\
& \phantom{+\delta_2} + x_d\theta_{23}\theta_{24}(\theta_{13}\theta_{23}^2 +\theta_{34}x_c) +
     \bigl(x_a x_d\theta_{13}\theta_{42} + x_b x_c\theta_{23}\theta_{14} + {\textstyle \prod\limits_{1\leq i<j\leq4}}\theta_{ij}\bigr)(\theta_{13}+\theta_{24})\!\Bigr)\! =0.
\end{split}
\end{align}

The above equation satisfies CAFCC in the form \eqref{CAFCCeqsA}, for the three values $(\delta_1,\delta_2)=(0,0),(1,0), (1,1)$.

\begin{align}\label{b2ddd}
    \begin{split}
B2&_{(\delta_1;\,\delta_2;\,\delta_3)}(x;x_a,x_b,x_c,x_d;\al,\bt)= \\ 
&    {\delta_1}\Bigl( \theta_{12}\theta_{43}(\theta_{12}^2-\theta_{13}\theta_{14}-\theta_{23}\theta_{24})^{\delta_2}(-x_a-x_b-x_c-x_d)^{\delta_3} \\
&\phantom{\delta} +\bigl(x_a(x+\theta_{14}\theta_{41}^{\delta_3})^{1+{\delta_2}}-x_b(x+\theta_{13}\theta_{31}^{\delta_3})^{1+{\delta_2}}-x_c(x+\theta_{24}\theta_{42}^{\delta_3})^{1+{\delta_2}}+x_d(x+\theta_{23}\theta_{32}^{\delta_3})^{1+{\delta_2}}\bigr)\!\Bigr)   \\
& + {\delta_3}\bigl((x_a x_c-x_b x_d)\theta_{12} + (x_a x_b-x_c x_d)\theta_{34} - x_bx_c(x_a+x_d) + x_ax_d(x_b+x_c)\bigr) \\
& + 2{\delta_2}( \theta_{12}\theta_{34})x^2 + (2{\delta_2} x +{\delta_3})\theta_{12}\theta_{34}(\theta_{13}+\theta_{24})+ (x_a x_d - x_b x_c) (\theta_{31}+\theta_{42})^{\delta_3}(-1)^{\delta_2} =0.
    \end{split}
\end{align}

\begin{align}\label{c2ddd}
    \begin{split}
&C2_{(\delta_1;\,\delta_2;\,\delta_3)}(x;x_a,x_b,x_c,x_d;\al,\bt)= \\ 
& (x_d -x_c)(x^2 + x_a x_b) + \theta_{34} (x^2 - x_a x_b)(\theta_{13}+\theta_{14})^{\delta_3} +2 {\delta_3}(\theta_{23} x_a - \theta_{24} x_b)x^2 \\
& + \Bigl((x_a + x_b + 2 {\delta_2} \theta_{23}\theta_{24})(x_c - x_d) - (x_a-x_b)(\theta_{23}+\theta_{24})(\theta_{13}+\theta_{14})^{\delta_3} + 2 {\delta_3} \theta_{34} x_a x_b \Bigr)x \\
& + {\delta_1}\bigl(x_a \theta_{24} -x_b \theta_{23} + \theta_{34}(\delta_2\theta_{23}\theta_{24}-x)\bigr)
\Bigl(\theta_{13}^{1+{\delta_2}+{\delta_3}}+\theta_{14}^{1+{\delta_2}+{\delta_3}} + 2 {\delta_2} x_c x_d -(x_c+x_d)(\theta_{13}+\theta_{14})^{\delta_2}\Bigr) \\
& +\delta_1 \theta_{23}\theta_{24}\bigl(x_c-x_d+\theta_{34}(\theta_{13}+\theta_{14}- 2 x)^{{\delta_3}}\bigr)(x_a+x_b-\theta_{34}^2-\theta_{23}\theta_{24})^{\delta_2} =0.
    \end{split}
\end{align}

The equations \eqref{a2dd}, \eqref{b2ddd}, \eqref{c2ddd}, collectively satisfy CAFCC in the form \eqref{CAFCCeqsBC}, with the four values $(\delta_1,\delta_2,\delta_3)=(0,0,0),(1,0,0),(1,0,1),(1,1,0)$.

\subsection{\texorpdfstring{$A2_{(\delta_1=0;\,\delta_2=0)}$, $D1$, $C1$}{A1(delta1=0;delta2=0), D1, C1}}\label{app:b1}

\begin{align}\label{d1}
 D1(x_a,x_b,x_c,x_d)= x_a-x_b-x_c+x_d=0,
\end{align}
 \begin{align}\label{c1}
\begin{split}
C1(x;x_a,x_b,x_c,x_d;\al,\bt)=\,
&(x_c+x_d)x^2 +  \bigl(2(\beta_1+\beta_2-2\alpha_2)-(x_a + x_b) (x_c + x_d)\bigr)x  \\
& +2 \bigl(\alpha_2(x_a+x_b) - \beta_2x_a - \beta_1x_b\bigr) + x_ax_b(x_c+x_d)=0.
\end{split}
\end{align}
The type-B equation \eqref{d1}, and type-C equation \eqref{c1}, along with the type-A equation \eqref{a2dd} with $(\delta_1,\delta_2)=(0,0)$, collectively satisfy CAFCC in the form \eqref{CAFCCeqsBC}.

\subsection{Four-leg expressions}

The above equations have equivalent four-leg type expressions that are given respectively in \eqref{4leg}--\eqref{4legc}, with the functions given in Tables \ref{table-A} and \ref{table-BC2} below.  The abbreviation {\it add.}, indicates an additive form of one of the equations \eqref{4leg}, \eqref{4legb}, \eqref{4legc}, given respectively by
\begin{align}
    a(x;x_a;\alpha_2,\beta_1)+a(x;x_d;\alpha_1,\beta_2)-a(x;x_b;\alpha_2,\beta_2)-a(x;x_c;\alpha_1,\beta_1)=0, \\
    b(x;x_a;\alpha_2,\beta_1)+b(x;x_d;\alpha_1,\beta_2)-b(x;x_b;\alpha_2,\beta_2)-b(x;x_c;\alpha_1,\beta_1)=0, \\
    a(x;x_a;\alpha_2,\beta_1)+c(x;x_d;\alpha_1,\beta_2)-a(x;x_b;\alpha_2,\beta_2)-c(x;x_c;\alpha_1,\beta_1)=0.
\end{align}

\begin{table}[htb!]
\centering
\begin{tabular}{c|c}

 Type-A & $a(x;y;\alpha,\beta)$ 
 
 \\
 
 \hline
 
 \\[-0.4cm]



$A3_{(\delta=1)}$ & $\displaystyle
\frac{\alpha^2+\beta^2\overline{x}^2-2\alpha\beta \overline{x}y}
     {\beta^2+\alpha^2\overline{x}^2-2\alpha\beta \overline{x}y}$

\\[0.4cm]

$A3_{(\delta=0)}$ & $\displaystyle
\frac{\beta x-\alpha y}{\alpha x-\beta y}$

\\[0.4cm]

$A2_{(\delta_1=1;\,\delta_2=1)}$ & $\displaystyle
\frac{(\sqrt{x}+\alpha-\beta)^2-y}{(\sqrt{x}-\alpha+\beta)^2-y}$

\\[0.4cm]

$A2_{(\delta_1=1;\,\delta_2=0)}$ & $\displaystyle
\frac{-x+y+\alpha-\beta}{x-y+\alpha-\beta}$

\\[0.4cm]

$A2_{(\delta_1=0;\,\delta_2=0)}$ & $\displaystyle
\phantom{(add.) }\quad  \frac{\alpha-\beta}{x-y} \quad (add.)$

\\[0.25cm]

\hline 
\end{tabular}
\caption{A list of the $a(x;y;\alpha,\beta)$ in \eqref{4leg} for type-A equations \eqref{a3d}, \eqref{a2dd}.  Here $\overline{x}=x+\sqrt{x^2-1}$.}
\label{table-A}
\end{table}


\begin{table}[htb!]
\centering
\begin{tabular}{c|c}

Type-B  & $b(x;y;\alpha,\beta)$ 
 
 \\
 
 \hline
 
 \\[-0.4cm]

$B3_{(\delta_1=\frac{1}{2};\,\delta_2=\frac{1}{2};\,\delta_3=0)}$ &  $\displaystyle
\beta^2+\alpha^2x^2-2\alpha\beta xy$

\\[0.2cm]

$B3_{(\delta_1=\frac{1}{2};\,\delta_2=0;\,\delta_3=\frac{1}{2})}$  & $\displaystyle
\frac{\alpha y-\beta\overline{x}}{\alpha\overline{x}y-\beta}$

\\[0.2cm]

$B3_{(\delta_1=1;\,\delta_2=0;\,\delta_3=0)}$ &  $\displaystyle
\beta-\alpha xy$

\\[0.2cm]

$B3_{(\delta_1=0;\,\delta_2=0;\,\delta_3=0)}$ ($D4$)  & $\displaystyle
y$

\\[0.2cm]

$B2_{(\delta_1=1;\,\delta_2=1;\,\delta_3=0)}$ & $\displaystyle
(x+\alpha-\beta)^2-y$

\\[0.2cm]

$B2_{(\delta_1=1;\,\delta_2=0;\,\delta_3=1)}$  & $\displaystyle
\frac{\sqrt{x}+y+\alpha-\beta}{-\sqrt{x}+y+\alpha-\beta}$

\\[0.2cm]

$B2_{(\delta_1=1;\,\delta_2=0;\,\delta_3=0)}$ & $\displaystyle
x+y+\alpha-\beta$

\\[0.2cm]

$B2_{(\delta_1=0;\,\delta_2=0;\,\delta_3=0)}$ ($D4$) & $\displaystyle
y$

\\[0.2cm]

$D1$  & $\displaystyle
\phantom{(add.) }\quad  y \quad (add.)$

\\[0.05cm]

\hline 
\end{tabular}
\hspace{0.0cm}
\begin{tabular}{c|c}

 Type-C & $c(x;y;\alpha,\beta)$ 
 
 \\
 
 \hline
 
 \\[-0.4cm]

 $C3_{(\delta_1=\frac{1}{2};\,\delta_2=\frac{1}{2};\,\delta_3=0)}$  
&
$\displaystyle
\frac{\alpha-\beta\overline{x}y}{\alpha\overline{x}-\beta y}$

\\[0.2cm]

$C3_{(\delta_1=\frac{1}{2};\,\delta_2=0;\,\delta_3=\frac{1}{2})}$
&
$\bigl(\tfrac{\alpha^2}{\beta}+\beta x^2-2\alpha x y\bigr)$

\\[0.2cm]

$C3_{(\delta_1=1;\,\delta_2=0;\,\delta_3=0)}$      
&
$xy-\tfrac{\alpha}{\beta}$

\\[0.2cm]

$C3_{(\delta_1=0;\,\delta_2=0;\,\delta_3=0)}$ 
&
$y$

\\[0.2cm]

$C2_{(\delta_1=1;\,\delta_2=1;\,\delta_3=0)}$
&
$\displaystyle
\frac{-\sqrt{x}+y-\alpha+\beta}{\sqrt{x}+y-\alpha+\beta}$

\\[0.2cm]

$C2_{(\delta_1=1;\,\delta_2=0;\,\delta_3=1)}$ 
&
$(x-\alpha+\beta)^2-y$

\\[0.2cm]

$C2_{(\delta_1=1;\,\delta_2=0;\,\delta_3=0)}$ 
&
$\displaystyle
x+y-\alpha+\beta$

\\[0.2cm]

$C2_{(\delta_1=0;\,\delta_2=0;\,\delta_3=0)}$ 
&
$\phantom{(a) }$ \,  $-\frac{y+\beta}{2x}  \quad (add.)$

\\[0.2cm]

$C1$ 
& 
$\phantom{(add.) }\quad  y \quad (add.)$

\\[0.05cm]

\hline 
\end{tabular}
\caption{Left: A list of the $b(x;y;\alpha,\beta)$ in \eqref{4legb} for the type-B equations \eqref{b3ddd}, \eqref{b2ddd}. Right: A list of the $c(x;y;\alpha,\beta)$ in \eqref{4legc} for the type-C equations \eqref{c3ddd}, \eqref{c2ddd}.  For $C3_{(\delta_1;\,\delta_2;\,\delta_3)}$ the $a(x;y;\alpha,\beta)$ is given by $A3_{(2\delta_2)}$, for $C2_{(\delta_1;\,\delta_2;\,\delta_3)}$ the $a(x;y;\alpha,\beta)$ is given by $A2_{(\delta_1;\,\delta_2)}$, and for $C1$ the $a(x;y;\alpha,\beta)$ is given by $A2_{(\delta_1=0;\,\delta_2=0)}$.  Here $\overline{x}=x+\sqrt{x^2-1}$.}
\label{table-BC2}
\end{table}



\end{appendices}

{

\bibliographystyle{utphys}

\providecommand{\href}[2]{#2}\begingroup\raggedright\endgroup

}

\end{document}